\documentclass[11pt, oneside]{article}   	
\usepackage{geometry}                		
\usepackage{graphicx}				
\usepackage{amssymb}
\usepackage{float}
\usepackage{amssymb}
\usepackage{amsmath}
\usepackage{amsthm}
\usepackage{color}
\usepackage{graphicx}
\usepackage{hyperref}
\usepackage{fullpage}
\usepackage{graphicx}
\usepackage{mathtools}
\usepackage{microtype}
\usepackage{bbm}
\usepackage[linesnumbered,algoruled,titlenotnumbered,vlined]{algorithm2e}
\DontPrintSemicolon
\usepackage{subfigure}
\usepackage{color}
\usepackage{algorithmic}

\usepackage{lmodern}
\usepackage[T1]{fontenc}
\usepackage[para,online,flushleft]{threeparttable}
\usepackage{natbib}

\newtheorem{theorem}{Theorem}
\newtheorem{lemma}[theorem]{Lemma}
\newtheorem{definition}[theorem]{Definition}

\newtheorem{example}[theorem]{Example}
\newtheorem{proposition}[theorem]{Proposition}
\newtheorem{assumption}[theorem]{Assumption}

\newcommand{\ie}{{\em i.e.}}

\newcommand{\poly}{{\em poly}}

\DeclareMathOperator*{\argmax}{argmax}
\DeclareMathOperator*{\argmin}{argmin}

\begin{document}
\title{A Faster FPTAS for Knapsack Problem With Cardinality Constraint}

\author{Wenxin Li \\
 Department of ECE\\
 The Ohio State University\\
{\tt wenxinliwx.1@gmail.com}\\
{\tt li.7328@osu.edu}
\and
Joohyun Lee\\
Division of Electrical Engineering\\
Hanyang University\\
{\tt joohyunlee@hanyang.ac.kr}
}
\begin{titlepage}

\maketitle

\begin{abstract}
We study the $K$-item knapsack problem (\ie, $1.5$-dimensional knapsack problem), a generalization of the famous 0-1 knapsack problem (\ie, $1$-dimensional knapsack problem) in which an upper bound $K$ is imposed on the number of items selected. This problem is of fundamental importance and is known to have a broad range of applications in various fields. It is well known that, there is no \emph{fully polynomial time approximation scheme} (FPTAS) for the $d$-dimensional knapsack problem when $d\geq 2$, unless P $=$ NP. While the $K$-item knapsack problem is known to admit an FPTAS, the complexity of all existing FPTASs have a high dependency on the cardinality bound $K$ and approximation error $\varepsilon$, which could result in inefficiencies especially when $K$ and $\varepsilon^{-1}$ increase. The current best results are due to [Mastrolilli and Hutter, 2006], in which two schemes are presented exhibiting a space-time tradeoff--one scheme with time complexity $O(n+Kz^{2}/\varepsilon^{2})$ and space complexity $O(n+z^{3}/\varepsilon)$, and another scheme that requires a run-time of $O(n+(Kz^{2}+z^{4})/\varepsilon^{2})$ but only needs $O(n+z^{2}/\varepsilon)$ space, where $z=\min\{K,1/\varepsilon\}$.

In this paper we close the space-time tradeoff exhibited in [Mastrolilli and Hutter, 2006] by designing a new FPTAS with a running time of $\widetilde{O}(n+z^{2}/\varepsilon^{2})$, while simultaneously reaching a  space complexity\footnote{$\widetilde{O}$ notation hides terms poly-logarithmic in $n$ and $1/\varepsilon$.} of $O(n+z^{2}/\varepsilon)$. Our scheme provides $\widetilde{O}(K)$ and $O(z)$ improvements on the state-of-the-art algorithms in time and space complexity respectively, and is the \emph{first} scheme that achieves a running time that is \emph{independent} of the cardinality bound $K$ (up to logarithmic factors) under fixed $\varepsilon$. Another salient feature of our algorithm is that it is the \emph{first} FPTAS that achieves better time and space complexity bounds than the very first standard FPTAS \emph{over all parameter regimes}.

%
%

\end{abstract}
\end{titlepage}


\section{Introduction}
The famous \emph{0-1 knapsack problem} (0-1 KP), also known as the \emph{binary knapsack problem} (BKP), is a classical combinatorial optimization problem which often arises when there are resources to be allocated within a budget. In addition, the 0-1 knapsack problem can be also viewed as the most fundamental non-trivial \emph{integer linear programming} (ILP) problem, and can be formally formulated as follows:
\begin{align}
 &\max{\sum_{i\in E}{p_{i} x_{i}}},\label{kpdef1}\\
s.t. &\sum_{i\in E}{w_{i}x_{i}}\leq W  \mbox{ and } x_{i}\in \{0,1\}.\label{kpdef2}
\end{align}
The value and size of each item $i$ is called \emph{profit} $(p_i)$ and \emph{weight} $(w_i$) respectively. For any positive integer $m$, let $[m]=\{1,2,\ldots, m\}$, we use set $E=[n]$ to denote the ground set, which includes all possible items. Our goal is to make a binary choice for each item $i$ to maximize the overall profit subject to a budget constraint $W$. Beyond this basic model, there are several extensions and variations of 0-1 KP, readers are referred to \cite{kellerer2003knapsack} for details.

In this paper, we study the \emph{$K$-item knapsack problem} ($K$KP), a well known generalization of the famous 0-1 KP that can be formulated as (\ref{kpdef1})-(\ref{kpdef2}) with the additional constraint $\sum_{i\in E}{x_{i}}\leq K$, which means that the number of items in any feasible solutions is upper bounded by $K$. The $K$KP can be cast as a special case of the \emph{two-dimensional knapsack problem}, which is a knapsack problem with two different packing constraints. Hence $K$KP problem can also be interpreted as \emph{$1.5$-dimensional knapsack problem} ($1.5$-KP)~\cite[p.~269]{kellerer2003knapsack}. Another closely related problem is the \emph{exact $K$-item knapsack problem} (E-$K$KP), for which the results in this paper still hold and discussions are included in \ref{appendixekkp}.

The $K$KP (and E-$K$KP) represents many practical applications in various fields ranging from \emph{assortment planning}~\cite{desir2016assortment} to \emph{multiprocessor task scheduling}~\cite{caprara2000approximation}, and \emph{crowdsourcing}~\cite{wu2015hear}. For example, the worker selection problem in crowdsourcing systems~\cite{Gong:2018:ITD:3209582.3209599}, \ie, maximizing opinion diversity in constructing a wise crowd, can be reduced to E-$K$KP. On the other hand, $K$KP also appears as a key subproblem in the solutions of several more complicated problems~\cite{aardal2015approximation,ahuja2004multi,epstein2012bin,jansen2006preemptiveresource,martello1999dynamic}. For example, in the \emph{bin packing} problem~\cite{epstein2012bin}, to apply the ellipsoid algorithm to approximately solve the linear program, the (approximation) algorithm to the $K$KP is utilized to construct a polynomial time (approximate) separation oracle. In many such practical and theoretical applications, 
the subroutine utilized to solve $K$KP frequently appears to be one of the main complexity bottleneck. These observations and facts motivate our study of designing a faster algorithm for $K$KP. 

\paragraph{Complexity of knapsack problems.} An FPTAS is highly desirable for NP-hard problems. Unfortunately, it has been shown that there exists no FPTAS for $d$-dimensional knapsack problem for $d\geq 2$, unless P$=$NP~\cite{magazine1984note}.

\subsection{Theoretical motivations and contributions}
\paragraph{Known results of $K$KP.} In this paper we focus on FPTAS for $K$KP (and E-$K$KP). The first FPTAS for $K$KP was proposed in \cite{caprara2000approximation}, by utilizing standard dynamic programming and profit scaling techniques, which runs in  $O(nK^{2}/\varepsilon)$ time and requires $O(n+K^{3}/\varepsilon)$ space. This algorithm was later improved by~\cite{mastrolilli2006hybrid}. Based on the \emph{hybrid rounding} technique, two alternative FPTASs (denoted by Scheme A and Scheme B) were presented, which significantly accelerate the dynamic programming procedure while exhibiting a space-time tradeoff. More specifically, Scheme A achieves a time complexity of $O(n+Kz^{2}/\varepsilon^{2})$ and space complexity of $O(n+z^{3}/\varepsilon)$, Scheme B needs $O(n+z^{2}/\varepsilon)$ space but requires a run-time of $O(n+(Kz^{2}+z^{4})/\varepsilon^{2})$. We remark that~\cite{krishnan2006multiscale} also investigated this problem, under an additional assumption that item profits follow an underlying distribution. This assumption enables the design of a fast algorithm via rounding the item profits adaptively according to the profit distribution.

The current fastest FPTAS (Scheme A) sacrifices its space complexity, in order to improve run-time performance. This may not be desirable as the space requirement is often a more serious bottleneck for practical applications than running time~\cite[p.~168]{kellerer2003knapsack}. Despite the recent widespread applications of the $K$KP problem~\cite{desir2016assortment,wu2015hear,ahuja2004multi, epstein2012bin,nobibon2011optimization,soldo2012optimal}, the state-of-the-art complexity results established in~\cite{mastrolilli2006hybrid} have not been improved since then. This lack of progress brings us to our first question: \emph{Is it  possible to design a more efficient FPTAS with lower time and/or space complexity to enhance practicality?}
	
Moreover, while the two schemes in~\cite{mastrolilli2006hybrid} achieve substantial improvements compared with~\cite{caprara2000approximation}, it is worth noting that there exists a hard parameter regime $\mathcal{H}=\{(n,K,\varepsilon)|K=\Theta(n),  \varepsilon^{-1}=\Omega(n)\}$, in which existing FPTASs in the literature fail to surpass both the time and space complexity barriers  guaranteed by the standard scheme in~\cite{caprara2000approximation}. For example, the run-time of Scheme B is higher than that of~\cite{caprara2000approximation}. Hence from a theoretical point of view, it is natural to ask: \emph{Can we design a new FPTAS that has lower time complexity or space complexity than the standard FPTAS~\cite{caprara2000approximation} over all parameter regimes?} 

\begin{table}[h]
\label{tableresult}
\centering
\begin{threeparttable}
{\begin{tabular}{|l|l|l|l|}\hline \textbf{Reference}& \textbf{Year}&\textbf{Time Complexity} & \textbf{Space Complexity}\\ \hline \hline \cite{caprara2000approximation}& $2000$ &$O(\frac{nK^{2}}{\varepsilon})$ & $O(n+\frac{K^{3}}{\varepsilon})$  \\\hline \cite{mastrolilli2006hybrid} (Scheme A) & $2006$ &$O(n+\frac{Kz^{2}}{\varepsilon^{2}})$&$O(n+\frac{z^{3}}{\varepsilon})$\\ \hline \cite{mastrolilli2006hybrid} (Scheme B)& $2006$ & $O(n+\frac{Kz^{2}+z^{4}}{\varepsilon^{2}})$& $O(n+\frac{z^{2}}{\varepsilon})$\\ \hline This Paper  &$2019$ &$\widetilde{O} (n + \frac{z^2}{\varepsilon^2})$ &$O(n+\frac{z^{2}}{\varepsilon})$\\ \hline\end{tabular}}	
\caption{Comparisons between different FPTASs. Here $z=\min\{K, \varepsilon^{-1}\}$, and as shown in Theorem~\ref{mainalgoguaran}, our time complexity can be refined to $\widetilde{O}(n+z^{4}+(z^{2}/\varepsilon)\cdot\min\{n,\varepsilon^{-1}\})$.}
\end{threeparttable}
\end{table}

\paragraph{Our contributions.} 
As summarized in Table~\ref{tableresult}, we break the longstanding barrier and answer the aforementioned questions in the affirmative. In particular, we present a new FPTAS with $\widetilde{O}(n + z^2/\varepsilon^2)$ running time and $O(n+z^{2}/\varepsilon)$ space requirement, which offers $\widetilde{O}(K)$ and $O(z)$ improvements in time and space complexity respectively. Our FPTAS is the \textbf{first} to achieve time complexity that is independent of $K$ (up to logarithmic factors, for a given $\varepsilon$). According to Theorem~\ref{mainalgoguaran}, the time complexity of our algorithm can be indeed refined to $\widetilde{O}(n+z^{4}+(z^{2}/\varepsilon)\cdot\min\{n,\varepsilon^{-1}\})$. From this refined bound, it can be seen that even in the hard regime $\mathcal{H}$, our algorithm has the same time complexity (up to $\log$ factors) as the standard FPTAS~\cite{caprara2000approximation}, while improving its space complexity by a factor of $n$. This implies that our algorithm is also the \textbf{first} FPTAS that outperforms the standard FPTAS~\cite{caprara2000approximation} over all parameter regimes, thus answering the second question in the affirmative.

Our new scheme also helps to improve the state-of-the-art complexity results of several problems in other fields, owing to the widespread applications of $K$KP (and E-$K$KP). In Appendix \ref{Secapp}, we take the resource constrained scheduling problem \cite{jansen2006preemptiveresource} as an illustrative example.

\subsection{Technique Overview}
Different from the \emph{hybrid rounding} technique proposed in \cite{mastrolilli2006hybrid}, which simplifies the structure of the input instance and approximately guarantees the objective value, we show that it is possible to achieve a better complexity result solely via \emph{geometric rounding} in the preprocessing phase. We divide items into two classes according to their profits and present distinct methods for each class of items. To solve the subproblem for items with low profit, we present a continuous relaxation function, using the natural linear programming relaxation and other alternatives based on structured weights and scaled budget constraint. The carefully designed relaxation function well approximates the optimal objective value of the subproblem and allows us to exploit the redundancy among various input. For every new input parameters, the relaxation can be computed in $O(z/\varepsilon)$ time on average. As for items with large profit, our treatment mainly follows from the novel ``functional'' approximation approach and point of view, which was recently proposed in~\cite{chan2018approximation}. As a straightforward generalization of the 0-1 KP, a two dimensional convolution operator is defined. We perform the convolution procedure in parallel planes to reduce the running time. The fact that there are at most $z$ elements with large profits helps us to bound the discretization precision via parameter $z$, instead of the number of profit functions. Here we adopt a slightly different but rather (unnecessary) sophisticated and tedious presentation via the lens of numerical discretization. We hope that this presentation helps to make the approach more clear (in the context of $K$KP). Finally, an approximate solution is obtained by appropriately putting these two modules together.

\section{Item Preprocessing}\label{overviewsec}

\begin{definition}[Item Partition]\label{classpartition}
Let $\mathcal{L}$ and $\mathcal{S}$ denote the set of large and small items, respectively. Item $e\in E$ is called a \emph{small} item if its profit is no more than $\varepsilon \mathrm{OPT}$, otherwise it is called a \emph{large} item\footnote{We discuss the method of obtaining $\mathrm{OPT}$ in \ref{appendixknowopt}.}, \ie, $\mathcal{S}=\{e\in E|\varepsilon \mathrm{OPT}/K\leq p_{e}\leq \varepsilon \mathrm{OPT}\}$ and $\mathcal{L}=\{e\in E| p_{e}\in \Xi\}$, where $\Xi=[\varepsilon \mathrm{OPT}, \mathrm{OPT}]$. We further divide $\mathcal{L}$ and $\mathcal{S}$ into different classes, $\{\mathcal{L}^{\dag}_{i}\}_{i\in [r_{\mathcal{L}}]}$ and $\{\mathcal{S}^{\dag}_{i}\}_{i\in [r_\mathcal{S}]}$, where $\mathcal{L}^{\dag}_{i}=\{e\in \mathcal{L}|p_{e}\in (\varepsilon(1+\varepsilon)^{i-1} \mathrm{OPT},\; \varepsilon(1+\varepsilon)^{i}\mathrm{OPT} ] \}	\; ( i \in  [r_{\mathcal{L}}])$ and $\mathcal{S}^{\dag}_{i}=\{e\in \mathcal{S}|p_{e}\in (\varepsilon(1+\varepsilon)^{-i}\mathrm{OPT},\; \varepsilon(1+\varepsilon)^{-i+1}\mathrm{OPT} ] \}	\; (i \in  [r_{\mathcal{S}}])$.
Let $r$ denote the number of non-empty classes in $E$, as shown in \ref{nonemptyclass}, we have
\begin{align}\label{rupperbound}
r=O(\min\{r_{\mathcal{L}}+r_{\mathcal{S}},n\})=O(\min\{\log(K/\varepsilon)/\varepsilon,n\})=\widetilde{O}(\min\{1/\varepsilon, n\}).
\end{align}
\end{definition}

\begin{definition}[Geometric Rounding]\label{profitsim} Without loss of generality, we can assume that elements in the same class have the same profit value. More specifically, we let $p_{e}=p^{\dag}_{i}=\varepsilon(1+\varepsilon)^{i} \mathrm{OPT} \;(\forall e\in \mathcal{L}_{i})$ and $p_{e}=p^{\ddagger}_{i}=\varepsilon(1+\varepsilon)^{-i} \mathrm{OPT}\;(\forall e\in \mathcal{S}_{i})$.
\end{definition}
The simplification in Definitions \ref{classpartition} and \ref{profitsim} does not hurt the solution since it will incur a loss of $O(\varepsilon \mathrm{OPT})$ in the objective value. Let $O^{*}$ denote the optimal solution, exploiting the simple structure of item profits after item partition and profit rounding, we are able to derive the following more fine-grained bound on $|O^{*}\cap \mathcal{L}|$ and the size of $\mathcal{S}$. Its proof is deferred to \ref{appendixpro3.1}.		
\begin{proposition}\label{upperboundsmall}
There are no more than $|O^{*}\cap \mathcal{L}|\leq z$ large items in the optimal solution set $O^{*}$. Without loss of generality, we can assume that the number of small items $|\mathcal{S}|=O(\min\{K\cdot\log(K/\varepsilon)/\varepsilon,n\})=\widetilde{O}(\min\{K/\varepsilon,n\})$.
\end{proposition}

\section{Algorithm for Large Items}\label{largeitemsubsection}
To approximately solve the $K$-item knapsack problem on ground set $E$, the first step of our approach is to divide this problem into two smaller $K$KP problems, which are defined on the large item set $\mathcal{L}$ and small item set $\mathcal{S}$ respectively. In this section we study the subproblem on $\mathcal{L}$, which is the same as the original problem, except that the ground set is substituted by $\mathcal{L}$ and the cardinality upper bound $k$ must be no less than $z$.
	
\subsection{An abstract algorithm based on convolution}\label{AnAbstractAlgorithm}
In the following we first define the profit function $\varphi_{(\cdot)}(\cdot,\cdot): 2^{\mathcal{L}}\times \mathbb{R}^{+} \times [z]\rightarrow \mathbb{R}^{+}$. From the definition we can see that $\varphi_{\mathcal{L}}(\omega,k)$ is equal to the optimal objective value of the subproblem considered in this section. 
\begin{definition}[Profit function~\cite{chan2018approximation}]\label{defProfitFunction} For any given set $T\subseteq E$, real number $\omega$, and integer $k$, $\varphi_{T}(\omega,k)$ is given by $\varphi_{T}(\omega,k)=\max\{\sum_{e\in T^{\prime}}{p_{e}}|\sum_{e\in T^{\prime}}{w_{e}}\leq \omega, |T^{\prime}|\leq k,T^{\prime}\subseteq T \subseteq E\}$, which denotes the optimal objective value of the $K$-item knapsack problem that is defined on set $T$, while the budget and cardinality are $\omega, k$ respectively.
\end{definition}

Our objective is to approximately compute matrix $\mathsf{Q}_{\mathcal{L}}=\{\varphi_{\mathcal{L}}(\omega,k)\}_{\omega\in X, k\in [z]}$, in which the value of $X$ will be specified in Section \ref{FastConvolutionAlgorithm}. This matrix plays an important role in our final item combination procedure, as we will show later in Section \ref{mainalgosec}. To compute the profit function efficiently, we introduce the following \emph{inverse weight function} $\phi_{(\cdot)}(\cdot,\cdot): 2^{\mathcal{L}}\times \Xi \times [z]\rightarrow \mathbb{R}^{+}$, which is one of the key ingredients in computing the profit function.
	
\begin{definition}[Inverse weight function]\label{defInverseWeightFunction}
For any given set $T\subseteq E$, real number $p$ and integer $k$, $\phi_{T}(p,k)$ is given by $\phi_{T}(p,k)=\min\{\sum_{e\in T^{\prime}}{w_{e}}|\sum_{e\in T^{\prime}}{p_{e}}\geq p, |T^{\prime}|\leq k, T^{\prime}\subseteq T\}$, which characterizes the minimum possible total weights under which there exists a subset of $T$ with total profit being no less than $p$ and cardinality no more than $k$.
\end{definition}
	
An immediate consequence of Definitions~\ref{defProfitFunction} and~\ref{defInverseWeightFunction} is that we can easily obtain the value of $\varphi_{\mathcal{L}}(\omega,k)$ based on $\phi$, \ie, via equation $\varphi_{\mathcal{L}}(\omega,k)=\sup\{p\in \mathbb{R}^{+}|\phi_{\mathcal{L}}(p,k)\leq \omega\}$. Therefore it suffices to derive the inverse weight function $\phi_{\mathcal{L}}(\cdot,\cdot)$ to compute $\mathsf{Q}_{\mathcal{L}}$.  
	


\begin{algorithm}[H]\label{computingphi}
\small
    \caption{Computing $\phi_{\mathcal{L}}(\cdot,\cdot)$}
    \label{highlevel}
\textbf{Input:} Partition scheme $\mathcal{L}=\cup_{i=1}^{\ell}{\mathcal{L}^{(i)}}$, Convolution operator $\otimes$;\\
\textbf{Output:} $\phi_{\mathcal{L}}(\cdot,\cdot)$\\
\For{$i=1$ to $\ell$}
{$\phi_{\cup_{j=1}^{i}{\mathcal{L}^{(j)}}}(\cdot,\cdot)\leftarrow (\phi_{\cup_{j=1}^{i-1}{\mathcal{L}^{(j)}}}\otimes \phi_{\mathcal{L}^{(i)}})(\cdot,\cdot)$;\\}
\textbf{Return $\phi_{\mathcal{L}}(\cdot,\cdot)$}
\end{algorithm} 

\paragraph{Algorithm for computing $\phi_{\mathcal{L}}(\cdot, \cdot)$.}

If we partition the large item set $\mathcal{L}$ into $\ell$ disjoint subsets as $\mathcal{L}=\cup_{i=1}^{\ell}{\mathcal{L}^{(i)}}$, then $\phi_{\mathcal{L}}$ can be computed by performing convolution operations sequentially. We specify the details in Algorithm \ref{computingphi} and the \emph{convolution operator} $\otimes$ is defined as follows.
	
\begin{definition}[Two dimensional convolution operator $\otimes$]\label{convolutiondeflemma}
For any two disjoint sets $S_{1}, S_{2} \subseteq E$, we use $(\phi_{S_{1}}\otimes \phi_{S_{2}})(\cdot,\cdot)$ to denote the convolution of functions $\phi_{S_{1}}(\cdot,\cdot)$ and $\phi_{S_{2}}(\cdot,\cdot)$, then it can be represented as,
\begin{align*}
(\phi_{S_{1}}\otimes \phi_{S_{2}})(p,k)&=\min\Big\{\phi_{S_{1}}(p_{1},k_{1})+\phi_{S_{2}}(p_{2},k_{2}) \Big|k_{1}+k_{2}\leq k, p_{1}+p_{2}\geq p \Big\} \\
&\equiv \phi_{S_{1}\cup S_{2}}(p,k).
\end{align*}
\end{definition}
Under this notation, function $\phi_{\mathcal{L}}(\cdot,\cdot)$ defined on $\mathcal{L}$ can be represented as $\phi_{\mathcal{L}}(p,k)=(\otimes_{i=1}^{\ell}{\phi_{\mathcal{L}^{(i)}}})(p,k)$. It is important to remark that
the algorithm is a rather general description of the convolution procedure, and the partition scheme should be further specified. Generally speaking, different partition schemes will induce different complexity results. For example, if we partition $\mathcal{L}$ into singletons, \ie, $\mathcal{L}^{(i)}=\{e_{i}\}$ and $\ell=|\mathcal{L}|$, then $\phi_{\mathcal{L}}(p,K)=(\otimes_{i=1}^{|\mathcal{L}|}	{\phi_{\{e_{i}\}}})(p,K)$. In this case, the algorithm is equivalent to the standard dynamic programming paradigm. In each stage we are in charge of making the decision of whether to include item $e_{i}$ or not. 

In this paper, we divide $\mathcal{L}$ in the same way as that in Definition \ref{classpartition}, \ie, $\mathcal{L}^{(i)}=\mathcal{L}^{\dag}_{i}, \forall i\in [r_{\mathcal{L}}]$.

\subsection{Discretizing the function domain}\label{NumericalDiscretization}
At the current stage, it is worth pointing out that in the convolution operation between inverse weight functions, the profit variable $p$ appears as a decision variable that varies continuously in $\Xi$. In addition, we are not able to obtain the closed form solution of the convolution operation analytically. The solution is to transform the problem into a computationally tractable one via discretization, then compute an (approximate) solution utilizing the computable version.

\paragraph{Discretizing the profit space.}
To implement the convolution in polynomial time, we discretize the interval $\Xi$ with the points $\{x_{i}\}_{ i\in [m]}$ as $X=\{x_{i}: \varepsilon \mathrm{OPT}=x_{1}<x_{2}<\ldots<x_{m-1}<x_{m}=\mathrm{OPT}\}\subseteq \Xi$. We denote the \emph{discretization parameter} of $X$ by \emph{discretization parameter} $\delta_{X}=\max_{1 \leq i \leq m-1} {\{x_{i+1}-x_{i}\}}$. To tackle the computational challenge induced by the continuity of profit $p$, we execute the convolution operation over the discrete functions that are defined on $X\times [z]$,
\begin{align}\label{discreteprofit}
(\phi_{S_{1}}\otimes \phi_{S_{2}})^{X}(p,k)=\min_{p_{1}, p_{2}\in X} & \Big\{\phi^{X}_{S_{1}}(p_{1},k_{1})+\phi^{X}_{S_{2}}(p_{2},k_{2}) \Big|k_{1}+k_{2}\leq k, p_{1}+p_{2}\geq p \Big\}.
\end{align}
More specifically, we start with functions $\phi^{X}_{\mathcal{L}^{(i)}}$, and compute $\phi^{X}_{\cup_{j=1}^{i}\mathcal{L}^{(j)}}$ iteratively until $\phi^{X}_{\mathcal{L}}$ is obtained. In general, function $\phi^{X}_{\cup_{i\in I}\mathcal{L}^{(i)}}(\cdot,\cdot)\equiv(\otimes_{i\in I}{\phi^{X}_{\mathcal{L}^{(i)}}})(\cdot,\cdot)$ for any $I\subseteq [\ell]$. The discrete profit function $\varphi^{X}_{S}(\cdot,\cdot)$ can also be recovered by its relation with the inverse weight function, \ie, $\varphi^{X}_{S}(\omega,k)=\max\{p \in X: \phi^{X}_{S}(p,k)\leq \omega\},\forall S\subseteq E$. 
	
\paragraph{Convergence behaviour of $\varphi^{X}(\cdot,\cdot)$.}	
We first show point-wise convergence of $\{\varphi^{X}_{(\cdot)}(\cdot,\cdot)\}_{X}$ towards $\varphi_{(\cdot)}(\cdot,\cdot)$ when $\delta_{X}$ goes to zero. It is worth pointing out that the straightforward intuition that convergence occurs if discretization is small, may not always hold. Indeed we can verify that the weight function $\phi^{X}$ may not converge to $\phi$ through the following example.

\begin{example}[$\phi^{X}$ does not converge to $\phi$]\label{conterexample}
Considering sets $S_{i}=\{e^{(i)}_{1},e^{(i)}_{2}\}\;(i=1,2)$, where the item profits and weights are given by $(p_{e^{(i)}_{1}},w_{e^{(i)}_{1}})=(\mathrm{OPT}/8,\omega/2)\;(i=1,2)$, $(p_{e^{(1)}_{2}},w_{e^{(1)}_{2}})=(\mathrm{OPT}/3,\omega/4)$, $(p_{e^{(2)}_{2}},w_{e^{(2)}_{2}})=(\mathrm{OPT}/6,\omega/4)$. According to Definition \ref{convolutiondeflemma}, we know that $\phi_{S_{1}\cup S_{2}}(\mathrm{OPT}/2,3)=w_{e^{(1)}_{2}}+w_{e^{(2)}_{2}}=\omega/2$. Let the discretization set $X_{d}=\Xi\cap\{i\cdot \frac{OPT}{2^{d}}|i\in [2^{d}]\}$, then it follows that the spacing $\delta_{X_{d}}\leq \frac{OPT}{2^{d}}$ and $\delta_{X_{d}}\rightarrow 0$ as $d\rightarrow \infty$. However, since $p_{e^{(i)}_{2}}\notin X_{d}\;(i=1,2)$, we have $\phi^{X}_{S_{1}\cup S_{2}}(\mathrm{OPT}/2,3)=\omega/2+w_{e^{(1)}_{2}}+w_{e^{(2)}_{2}}=\omega\neq \phi_{S_{1}\cup S_{2}}(\mathrm{OPT}/2,3)$ and $\phi^{X}_{S_{1}\cup S_{2}}$ does not converge to $\phi_{S_{1}\cup S_{2}}$.
\end{example}


\begin{lemma}\label{convergelemma}
For any finite index set $I$ and $\omega,k$, we have $\lim_{\delta_{X}\rightarrow 0}{\varphi^{X}_{\cup_{i\in I}\mathcal{L}^{(i)}}(\omega,k)} =\varphi_{\cup_{i\in I}\mathcal{L}^{(i)}}(\omega,k)$ for fixed $\omega,k$.
\end{lemma}
\begin{proof}
It suffices to prove the case when $|I|=2$, because for the case when $|I|>2$, convergence can be proven by induction, using the result we have for $|I|=2$. When there are only two elements in $I$, it is easy to check that $\varphi_{\cup_{i\in I}\mathcal{L}^{(i)}}(\omega,k)-2\delta_{X}\leq \varphi^{X}_{\cup_{i\in I}\mathcal{L}^{(i)}}(\omega,k) \leq \varphi_{\cup_{i\in I}\mathcal{L}^{(i)}}(\omega,k)$, thus the proof is complete.	
\end{proof}

The theoretical convergence of $\varphi^{X}(\cdot,\cdot)$ ensures the near-optimality of the solution obtained by discretization, as long as $X$ is dense enough in $\Xi$. However, what matters greatly is the \emph{order of the accuracy}, which refers to how rapidly the error decreases in the limit as the discretization parameter tends to zero. The formal definition of the convergence speed of discretization methods is given as following.

\begin{definition}[\citep{michelle2002numerical}] 
Let $n$ be the number of grid points in the discretization process, the discretization method is said to converge with order $p$ if for the relevant sequence $\{x_{n}\}_{n\geq 0}$, there exists $L$ such that $|x_{n}-L|=O(n^{-p})$ holds.
\end{definition}

This speed is directly related to the complexity of our algorithm. From the following lemma, we can conclude that the method of discretizing $X$ by a uniform grid set converges with order $1$, as $\delta_{X}=O(1/|X|)$ for uniform grid set. 
\begin{lemma}\label{continuityvarphi}
Let $\phi^{X}_{\mathcal{L}}$ be the weight function, then for any given budget $\omega\leq W$, cardinality upper bound $k\leq z$, and discretization set $X$, we have $|\varphi^{X}_{\mathcal{L}}(\omega,k)-\varphi_{\mathcal{L}}(\omega,k)|\leq C\delta_{X}$, where the coefficient $C=z+1$. As a consequence, $|X|$ must be of order $\Omega(z/\varepsilon)$ to ensure an error of order $O(\varepsilon \mathrm{OPT})$.
\end{lemma}
\begin{proof}
The proof is deferred to Appendix \ref{appendixlemten}.
\end{proof}

\subsection{Fast convolution algorithm}\label{FastConvolutionAlgorithm}
Now we settle the problem of designing a fast convolution algorithm, which is the last remaining issue that has a critical impact on the efficiency of the algorithm for large items. To this end, we show an inherent connection between convolution results under different inputs $p$ and $k$, which is formally described in Lemma \ref{monoineq}. Owing to this observation, we are able to remove a large amount of redundant calculations when facing new input parameters. To start with, we first sort items in each $\mathcal{L}^{\dag}_{i}$ in non-increasing order of weights, which takes $O(z\log z)$ time. We define the optimum index function as follows.
	
\begin{definition}[Optimum index function]\label{optimalindex}
$\psi: X\times [K]\rightarrow [K]$ is defined as,
\begin{align}\label{indexfunction}
\psi(p,k)=\argmin &\Big\{\theta \in [k]\Big|\phi^{X}_{\mathcal{L}^{\dag}_{a}}(\max\{x\in X: x\leq \theta\cdot p^{\dag}_{a}\}, \theta)\notag\\
&+\phi^{X}_{S}(\max\{x\in X: x\leq p-\theta\cdot p^{\dag}_{a}\}, k-\theta)\Big\}\;(p\in X).
\end{align}	
\end{definition}
	
Here (\ref{indexfunction}) benefits from the partition in which all items in the same set $\mathcal{L}^{\dag}_{i}$ have equal profit value. Specifically, when we derive the result of  $(\phi^{X}_{\mathcal{L}^{\dag}_{a}}\otimes \phi^{X}_{S})(p,k)$, there is indeed only one decision variable $\theta$, \ie, the number of elements selected from $\mathcal{L}^{\dag}_{i}$, that should be figured out. Hence, we denote the optimal value of $\theta$ by the index function $\psi$. Our primary objective is then reduced to figure out all the indices $\{\psi(p,k)\}_{p\in X,k\in [z]}$, for which we give a graphic illustration in Figure~\ref{figsearchingspace}. It can be regarded as finding \emph{column minimums} in the cube, here column minimum refers to the optimal indices defined in Definition \ref{optimalindex}. 

\paragraph{Consider the problem in parallel slices.}
As shown in Figure \ref{ddcslice}, we divide the cube into parallel slices. Consider slice
\begin{align}
H=\Big\{(p,k)\Big|p=p_{0}+\zeta\lambda_{a},k=k_{0}+\zeta\Big\}\bigcap \Big( \Xi\times[0,z]	\Big),\label{slicedef}
\end{align}
where $(p_{0},k_{0})$ denotes the boundary point of slice $H$ and hence $p_{0}k_{0}=0$, $\zeta$ represents the drift of point $(p,k)$ from boundary. It can be seen that the angle between slice $H$ and the frontal plane is equal to $\arctan \lambda_{a}^{-1}$, and there are $O(|X|)=O(z/\varepsilon)$ such parallel slices in the cube. On the other hand, plugging (\ref{slicedef}) into (\ref{indexfunction}), the index function can be simplified to 
\begin{align*}
\chi_{H}(\zeta)=\argmin \Big\{ \theta\in [z]\Big|\phi^{X}_{\mathcal{L}^{\dag}_{a}}(\lambda_{a}\theta, \theta)+\phi^{X}_{S}(p_{0}+\lambda_{a}[\zeta-\theta], k_{0}+[\zeta-\theta])\Big\}.
\end{align*}
	
Without loss of generality we could assume that there exists an integer $\tau_{a}\in \mathbb{Z}^{+}$ such that $p^{\dag}_{a}=\tau_{a}\cdot \frac{\varepsilon \mathrm{OPT}}{z}$, otherwise we can always modify $p^{\dag}_{a}$ by an $O(\frac{\varepsilon \mathrm{OPT}}{z})$ additive factor to meet this criteria while inducing a $O(\varepsilon \mathrm{OPT})$ loss in the objective function. Consequently we have $\lambda_{a}=\tau_{a}\varepsilon \mathrm{OPT}$. We consider the case when $\Xi$ is discretized by the uniform grid set $X=\{i\cdot \frac{\varepsilon \mathrm{OPT}}{z} |i\in [z/\varepsilon]\}$. Then the following key observation about the distribution of column minima in slice $H$ holds. 
\begin{lemma}\label{monoineq} For any two columns in $H$ that are indexed by $\zeta_{1}$ and $\zeta_{2}$, we have
\begin{align}\label{gradientineq}
\frac{\chi_{H}(\zeta_{2})-\chi_{H}(\zeta_{1})}{\zeta_{2}-\zeta_{1}}\leq 1.
\end{align}
\end{lemma}
\begin{proof}
The proof is deferred to Appendix \ref{appendixlem12}.
\end{proof}

\paragraph{Divide-and-Conquer on slice $H$.}	In Lemma \ref{monoineq}, we establish an upper bound on the growth rate of the index function. Taking advantage of this lemma, we are able to reduce the size of the searching space in one column, given that we have figured out the optimum indices at some other columns in the slice $H$. More specifically, consider columns indexed by $\zeta_{1}\leq \zeta_{2}\leq \zeta_{3}$, the information of $\chi_{H}(\zeta_{1})$ and $\chi_{H}(\zeta_{3})$ indeed provide two cutting planes to help us locate $\chi_{H}(\zeta_{2})$ in a smaller interval $[\chi_{H}(\zeta_{3})+\zeta_{2}-\zeta_{3},\chi_{H}(\zeta_{1})+\zeta_{2}-\zeta_{1}]$.

Inspired by this observation, we design a \emph{divide-and-conquer} procedure to compute the optimum indices efficiently for any slice in the form of (\ref{slicedef}). We start with a recursive call to determine the optimum indices of all the even-indexed columns. Here a column is called even (odd) column if and only if its corresponding $\zeta$ value in (\ref{slicedef}) is even (odd). Then for each odd column $\chi_{H}(2i)$, it can be computed by enumerating the interval $[\chi_{H}(2i+1)-1, \chi_{H}(2i-1)+1]$. The details are specified in Appendix \ref{appendixsliceindex}. 
%

The time complexity of computing the index function for a single slice is summarized in the following proposition.
\begin{proposition}\label{sliceddc}
It takes $O(z\log z)=\widetilde{O}(z)$ time to compute $\chi_{H}(\cdot)$.	
\end{proposition}
\begin{proof}
See Appendix \ref{appendixsliceddc}.
\end{proof}

\begin{figure}[H]
\centering 
\subfigure[Searching Space of $\phi^{X}_{\mathcal{L}^{\dag}_{a}}\otimes \phi^{X}_{S}$]{
\label{figsearchingspace}
\includegraphics[width=0.45\textwidth]{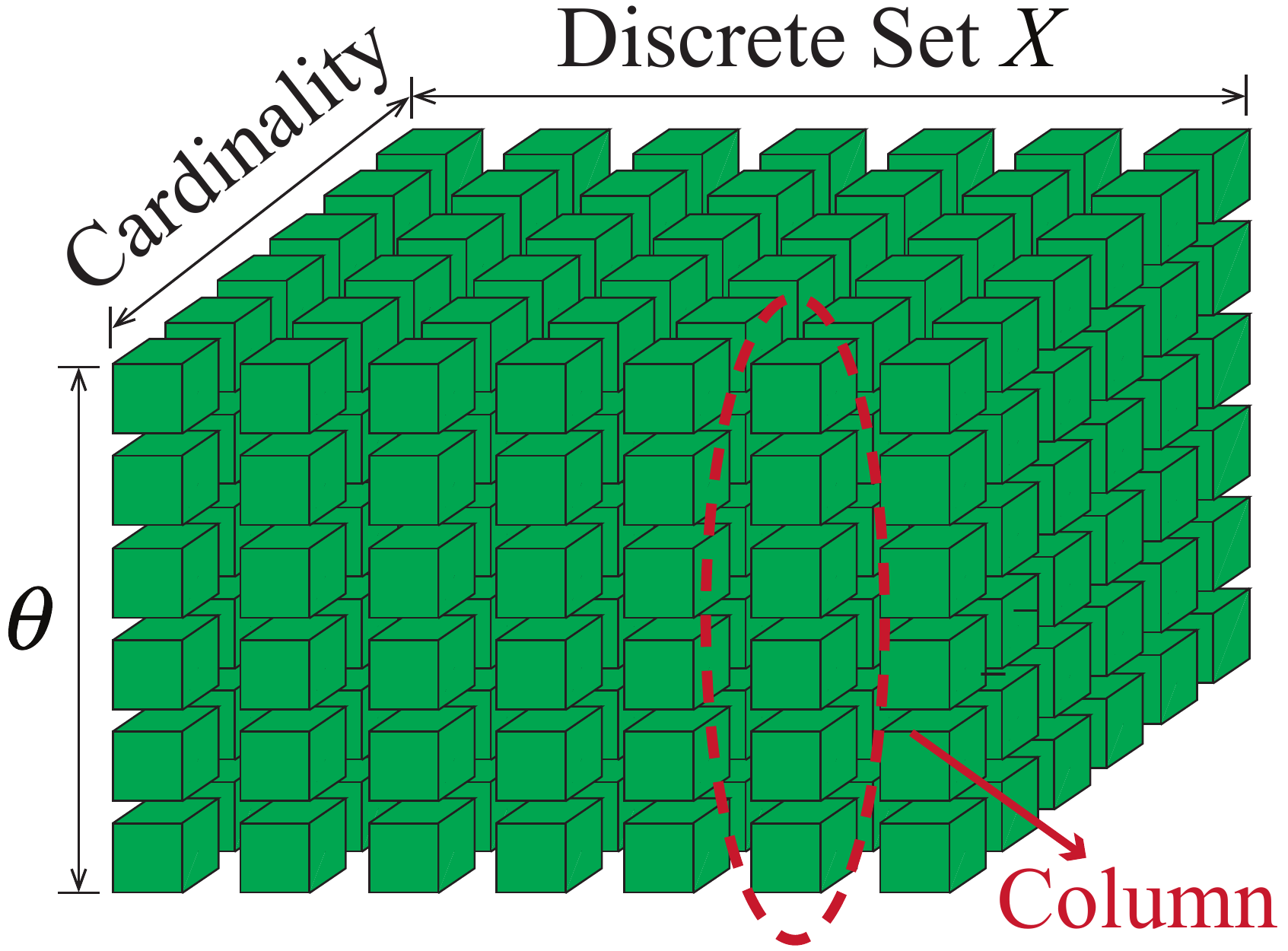}}
\subfigure[Divide-and-Conquer on $H$]{
\label{ddcslice}
\includegraphics[width=0.4\textwidth]{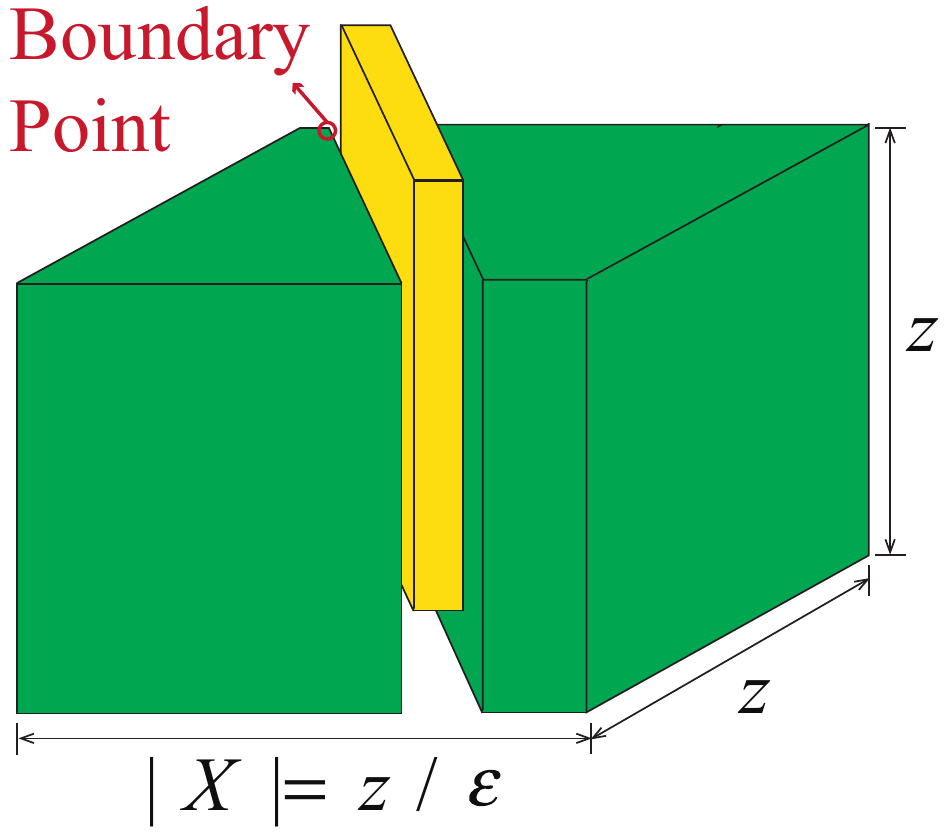}}
\label{Fig.main}
\caption{Graphic Illustrations of the Convolution Operation}
\end{figure}

\paragraph{Fast convolution operation.}
We are ready to introduce the convolution algorithm, using the concepts and algorithms developed in the previous subsections. Details are specified in Appendix \ref{appendixconvolution}. When the convolution operation is specified as Algorithm~\ref{convolution}, the complexities of Algorithm~\ref{highlevel} is presented as follows.

\begin{lemma}
\label{largetimecomplexity}
It takes $O(n)$ space and 
\begin{align*}
O(n+(z^{2}\log z/\varepsilon)\cdot \min \{ \log(1/\varepsilon)/\varepsilon, n\})=O(n)+\widetilde{O}(\min\{z^{2}/\varepsilon^{2},nz^{2}/\varepsilon\})
\end{align*}
time to complete the convolution operation.
\end{lemma}
\begin{proof}
See Appendix \ref{appendixle}.
\end{proof}

Generally speaking, for given $p\in X$ and $k\in \mathbb{Z}^{+}$, it requires $O(z\cdot|X|)$ arithmetic operations to compute $(\phi_{S_{1}}\otimes \phi_{S_{2}})(p,k)$, if we enumerate all possible pairs of $(p_{1},k_{1})$ in equation (\ref{discreteprofit}), which further results in a total complexity of $O(z^{2}|X|^{2})$ for operator $\otimes$. Compared with our Algorithm, this is unnecessarily inefficient, since it restarts all the arithmetic operations when the input parameters varies.

\section{Continuous Relaxation for Small Items}\label{smallputting}

In Section~\ref{largeitemsubsection} we have shown how to approximately select the most profitable large items under any given budget and cardinality constraints. One important task left is to solve the subproblem with only small items involved. In this section we show how to approximately solve this subproblem efficiently. Similar to Definition \ref{defInverseWeightFunction}, the profit function of small items, $\varphi_{\mathcal{S}}(\cdot,\cdot): \mathbb{R}^{+} \times [K]\rightarrow \mathbb{R}^{+}$, is given by $\varphi_{\mathcal{S}}(\omega,k)=\max\{\sum_{e\in \mathcal{S}}{p_{e}x_{e}}|\sum_{e\in \mathcal{S}}{x_{e}}\leq k, \sum_{e\in \mathcal{S}}{w_{e}x_{e}}\leq \omega,x_{e}\in \{0,1\} \}$.
The main spirit of our approach for small items is similar to that of Section \ref{largeitemsubsection}, \ie, find a new function $\varphi^{\dag}_{\mathcal{S}}$, which is a good approximation of $\varphi_{\mathcal{S}}$ and is economical in computations. 

One question that may arise is the following: can the methods in Section \ref{largeitemsubsection} still work for the small item set $\mathcal{S}$, \ie, can we apply Algorithm \ref{highlevel} over $\mathcal{S}$ and use the output discrete function as an approximation of $\varphi_{\mathcal{S}}$? It can be verified that $O(n)+\widetilde{O}(K^{2}/\varepsilon^{2})$ time is required, which is significantly high especially when $K$ is large, and fails to provide the desired complexity result. This is because that there could be many more small items than large items, which will result in a larger searching space. 

To construct the approximation function $\varphi^{\dag}_{\mathcal{S}}$, we turn to the continuous relaxation of the subproblem, as the continuous optimization problem is much easier to deal with. More importantly, the boundness of small item profits will ensure that the gap between optimal values of the two problems is sufficiently small. 

Our main result in this section is formally stated in the following Theorem \ref{continuousrelaxationmainlemma}. In the remaining of this section, we will show the correctness of Theorem \ref{continuousrelaxationmainlemma} step by step. We first present the details of $\varphi^{\dag}_{\mathcal{S}}$ and prove its approximation error in Section \ref{continuouserror}, then show the computational complexity of calculating $\varphi_{S}^{\dag}$ in Section \ref{computrelaxation}. The proof is summarized in Appendix  ~\ref{appendixcontinuousrelaxationmainlemma}. 
\begin{theorem}
\label{continuousrelaxationmainlemma}
There exists a function $\varphi^{\dag}_{\mathcal{S}}(\cdot,\cdot): \mathbb{R}^{+} \times [K]\rightarrow \mathbb{R}^{+}$, such that 
$|\varphi^{\dag}_{\mathcal{S}}(\omega, k)-\varphi_{\mathcal{S}}(\omega, k)|=O(\varepsilon \mathrm{OPT})$ for any $\omega$ and $k$. In addition, for any weight set $\mathcal{W}$ of size $O(1/\varepsilon)$ and cardinality bound set $\mathcal{K}$ of size $O(z)$, $\mathsf{Q}_{\mathcal{S}}=\{\varphi^{\dag}_{\mathcal{S}}(\omega,k)| \omega \in \mathcal{W}, k\in \mathcal{K}\}$ can be computed within  $\widetilde{O}(n+z^{4}+\min\{z^{2}/\varepsilon^{2},nz/\varepsilon\})$ time, while requiring $O(z/\varepsilon)$ space.
\end{theorem}

\subsection{Relaxation function design and approximation error analysis}\label{continuouserror}
We first introduce the two building block functions in our construction.
\begin{definition}[Definition of $\Upsilon_{1}(\cdot,\cdot)$ and $\Upsilon_{2}(\cdot,\cdot)$]\label{defcontin}
For any weight $\omega$ and cardinality $\ell$, function $\Upsilon_{1}(\cdot,\cdot): \mathbb{R}^{+} \times [K]\rightarrow \mathbb{R}^{+}\;(i=1,2)$ is defined as
\begin{align}\label{linrelaxation}
\Upsilon_{1}(\omega, k)=\max\Big\{\sum_{e\in \mathcal{S}}{p_{e}x_{e}}\Big|\sum_{e\in \mathcal{S}}{w_{e}x_{e}}\leq \omega, \sum_{e\in \mathcal{S}}{x_{e}}\leq k, x_{e}\in[0,1] \Big\}.
\end{align}	
The second relaxation function is constructed as the maximum summation of functions $\Upsilon_{3}$ and $\Upsilon_{4}$,
\begin{align}\label{secondrelax}
\Upsilon_{2}(\omega, k)=\max_{0\leq \ell \leq k}\{\Upsilon_{3}(\omega,\ell)+\Upsilon_{4}(\omega,\ell, k)\},
\end{align}
where $\Upsilon_{3}(\cdot,\cdot): \mathbb{R}^{+} \times [K]\rightarrow \mathbb{R}^{+}$ and $\Upsilon_{4}(\cdot,\cdot,\cdot): \mathbb{R}^{+} \times [K]  \times [K] \rightarrow \mathbb{R}^{+}$ are given by
\begin{align}\label{upsilon3}
\Upsilon_{3}(\omega, \ell)=\max\Big\{\sum_{e\in T}p_{e}\Big|T\subseteq \mathcal{S}_{1}(\omega), |T|\leq \ell \Big\},
\end{align}
\begin{align}
\label{varphi2tdef}
\Upsilon_{4}(\omega, \ell, k)= \max_{x_{e}\in [0, 1]}\Big\{\sum_{e\in \mathcal{S}_{2}(\omega)}{p_{e}x_{e}}\Big|\sum_{e\in \mathcal{S}_{2}(\omega)}{x_{e}}\leq k-\ell, \sum_{e\in \mathcal{S}_{2}(\omega)}{\bar{w}_{e}\cdot x_{e}}\leq (1-\varepsilon)\cdot\omega \Big\}.
\end{align}
Here $\mathcal{S}_{1}(\omega)=\{e\in\mathcal{S}|\omega_{e}\leq \varepsilon \omega/K\}$ represents the set of elements in $\mathcal{S}$ with weight less than threshold $\varepsilon \omega/K$, set $\mathcal{S}_{2}(\omega)=\{e\in \mathcal{S}|w_{e}\leq\omega\}\setminus \mathcal{S}_{1}(\omega)$. The modified weight $\bar{\omega}_{e}$ in (\ref{varphi2tdef}) is given as $\bar{w}_{e}= \omega_{e} \cdot (1+\varepsilon)^{\lceil \log_{(1+\varepsilon)}{(\frac{\varepsilon K w_{e}}{\omega})}\rceil}/(K\varepsilon)$, where $\lceil \cdot \rceil$ refers to the ceiling function.	

\end{definition}
	
The first function $\Upsilon_{1}(\omega, k)$ is the most natural linear programming relaxation of $\varphi^{\dag}_{\mathcal{S}}$, where all the integer variables are relaxed to real numbers in $[0,1]$. In the second function $\Upsilon_{2}(\omega, k)$, we only relax variables corresponding to elements in $\mathcal{S}_{\omega}$, while element weights are rounded to integer powers of $(1+\varepsilon)$, and the budget $\omega$ is scaled by a factor of $(1-\varepsilon)$.
	
In our algorithm, we let the approximation function
\begin{align*}
\varphi^{\dag}_{\mathcal{S}}(\omega,k)=\Upsilon_{1}(\omega,k)\cdot \mathbbm{1}_{\{K\leq \varepsilon^{-1}\}}+\Upsilon_{2}(\omega, k)\cdot \mathbbm{1}_{\{K> \varepsilon^{-1}\}}.
\end{align*}	
The following lemma shows that $\varphi^{\dag}_{\mathcal{S}}$ provides a good approximation of $\varphi_{\mathcal{S}}$.
\begin{lemma}\label{lemmaapproxi}
The differences between functions $\varphi^{\dag}_{\mathcal{S}}$ and $\varphi_{\mathcal{S}}$ is bounded as $|\varphi^{\dag}_{\mathcal{S}}(\omega,k)-\varphi_{\mathcal{S}}(\omega,k)|\leq 4\varepsilon \cdot\mathrm{OPT}$.
\end{lemma}
\begin{proof}
See Appendix \ref{appendixlemmaapproxi}.
\end{proof}
	
\subsection{Computing $\varphi^{\dag}_{\mathcal{S}}$ efficiently}\label{computrelaxation}
In this subsection, we consider how to compute set $\{\varphi^{\dag}_{\mathcal{S}}(\omega,k)| \omega \in \mathcal{W}, k\in \mathcal{K}\}$ efficiently, for any given $\mathcal{K}\in \mathbb{Z}^{|\mathcal{K}|}$ and $\mathcal{W}\in \mathbb{R}^{|\mathcal{W}|}$. We treat functions $\Upsilon_{3}$ and $\Upsilon_{4}$ separately. To compute relaxation $\Upsilon_{3}(\cdot,\cdot)$, it is worth noting that one straightforward approach is to utilize the linear time algorithm~\citep{megiddo1984linear, megiddo1993linear, caprara2000approximation} to solve equation (\ref{upsilon3}), for each pair of distinct parameters in $\mathcal{W}$ and $\mathcal{K}$. This will result in a total complexity of $O(|\mathcal{S}|\cdot |\mathcal{K}|\cdot |\mathcal{W}|)=O((z/\varepsilon)\cdot\min\{K/\varepsilon, n\})$, which has a high dependence on the parameter $K$.


\subsubsection{Computing relaxation $\Upsilon_{2}(\cdot,\cdot)$}	
For notational convience, we let $\Upsilon_{5}(\omega, \ell, k)=\Upsilon_{3}(\omega,\ell)+\Upsilon_{4}(\omega,\ell, k)$, then $\Upsilon_{2}(\omega, k)=\max_{0\leq \ell \leq k}{\Upsilon_{5}(\omega, \ell, k)}$ according to Definition~\ref{defcontin}. We first claim the following observation with regard to $\{\Upsilon_{5}(\omega, \ell, k)\}_{0\leq \ell\leq k}$, which enables us to compute $\Upsilon_{2}(\omega, k)$ for each fixed value of $\omega$ and $k$, via $O(\log k)$ calls to the routine of computing $\Upsilon_{5}(\omega, \ell, k)$.

\begin{lemma}[Concavity of $\{\Upsilon_{5}(\omega, \ell, k)\}_{0\leq \ell\leq k}$]\label{concavesequence}
Sequence $\{\Upsilon_{5}(\omega, \ell, k)\}_{0\leq \ell\leq k}$ is concave with respect to $\ell$, \ie, $\Upsilon_{5}(\omega, \ell_{1}, k)+\Upsilon_{5}(\omega, \ell_{2}, k)\leq 2\Upsilon_{5}(\omega, (\ell_{1}+\ell_{2})/2, k)$. Consequently, $\Upsilon_{2}(\omega, k)$ can be computed in $O(\mathcal{T}_{f}\cdot \log k)$ time, where $\mathcal{T}_{f}$ represents the worst case running time of computing $\Upsilon_{5}(\omega, \ell, k)$ under fixed values of $\omega, \ell, k$. 
\end{lemma}

\begin{proof}
We first show that $\{\Upsilon_{3}(\omega,\ell)\}_{\ell\in [k]}$ is a concave sequence. Note that the first order difference 
$\Delta \Upsilon_{3}(\omega,\ell)=\Upsilon_{3}(\omega,\ell)-\Upsilon_{3}(\omega,\ell-1)$, which is equal to the $\ell$-th largest profit in $\mathcal{S}_{1}(\omega)$. Thus, the first order sequence $\{\Delta \Upsilon_{3}(\omega,\ell)\}_{\ell \in[k]}$ is non-increasing and concavity of $\{\Upsilon_{3}(\omega,\ell)\}_{\ell\in [k]}$ follows. For sequence $\{\Upsilon_{4}(\omega,\ell, k)\}_{\ell\in [k]}$, we use $\mathbf{x}^{*}_{\omega, \ell, k}$ to denote the optimal fractional solution to (\ref{varphi2tdef}), $\Upsilon_{4}(\omega,\ell, k)\geq [\Upsilon_{4}(\omega,\ell-1, k)+\Upsilon_{4}(\omega,\ell+1, k)]/2$ holds, since $(\mathbf{x}^{*}_{\omega, k-\ell+1, k}+\mathbf{x}^{*}_{\omega, k-\ell-1, k})/2$ is a feasible solution to (\ref{varphi2tdef}) under cardinality bound $k-\ell$. The concavity of $\Upsilon_{5}$ follows from the fact that sequence concavity is preserved under summation.
		
As for the time complexity, notice that sequence concavity implies monotonicity of the first order difference sequence $\{\Delta \Upsilon_{5}(\omega, \ell, k)\}_{0\leq \ell\leq k}$. Hence  $\ell^{*}=\argmax_{\ell\in [k]}{\Upsilon_{5}(\omega, \ell, k)}$ can be derived via binary search, using the sign of $\Delta \Upsilon_{5}(\omega, \ell, k)$ as indication information. For each fixed $\omega$ and $\ell$, $\Upsilon_{2}(\omega, \ell)$ can be computed in $O(\mathcal{T}_{f}\cdot \log k)$ time. The proof is complete.
\end{proof}

\paragraph{Computing $\Upsilon_{5}(\omega, \ell, k)$} At the current stage, we have shown that $\Upsilon_{2}(\omega, k)$ can be computed within the same order of time (up to a factor of $O(\log k)$) as computing  $\Upsilon_{5}(\omega, \ell, k)$. In the following, we present our two subroutines of calculating  $\Upsilon_{3}(\omega, \ell)$ and $\Upsilon_{4}(\omega, \ell, k)$.

\begin{itemize}	
\item \emph{Calculating $\Upsilon_{3}(\omega, \ell)$}. Let weight set $\mathcal{W}=\{\omega_{1}\leq \omega_{2}\cdots \leq \omega_{|\mathcal{W}|}\}$. We partition and store the small item set $\mathcal{S}=\cup_{i=1}^{|\mathcal{W}|-1}{\mathcal{S}_{i}}$, where $\mathcal{S}_{1}=\{e\in \mathcal{S}|w_{e}\in [\omega_{1},\omega_{2}]\}$ and for $i \in [2, |\mathcal{W}|-1]$,
\begin{align}
\mathcal{S}_{i}=\Big\{e\in \mathcal{S}\Big|w_{e}\in (\omega_{i},\omega_{i+1}]\Big\},
\end{align}
which takes $O(|\mathcal{S}|\cdot \log |\mathcal{W}|)=\widetilde{O}(|\mathcal{S}|)$ time and $O(|\mathcal{S}|)$ space. Without loss of generality, we assume that items in sets $\mathcal{S}=\{e_{1},\ldots,e_{|\mathcal{S}|}\}$ and $\mathcal{S}_{i}=\{ e_{\iota_{i,1}}, e_{\iota_{i,2}},\cdots,e_{\iota_{i, |\mathcal{S}_{i}|}}\}$ are in non-increasing order of profit values, as sorting takes $\widetilde{O}(|\mathcal{S}|)$ time, which is a lower order term. We further store the partial summation sequence $\{\pi(i, j)\}_{i \in [|\mathcal{W}|-1], j\in [|\mathcal{S}_{i}|]}$, where 
\begin{align}\label{partitalsum}
\pi(i,j)=\sum_{t=1}^{j}{p_{e_{\iota_{i,t}}}}\end{align}
represents the total profits of the first $j$ items in $\mathcal{S}_{i}$. This procedure contributes a lower order term of $O(n)$ to the time and space complexity.

To compute function $\Upsilon_{3}(\omega_{i}, \ell)$, we first figure out the index of the $\ell$-th largest item in $\cup_{j \in [i]}{\mathcal{S}_{j}}$, again using binary search. Then $\Upsilon_{3}(\omega, \ell)$ can be computed based on the pre-computed partial summations defined in (\ref{partitalsum}), which takes $O(|\mathcal{W}|)$ time, under given values of $\omega$ and $\ell$. Hence the total complexity of computing $\Upsilon_{3}(\omega, \ell)$ for $\omega\in \mathcal{W}$ and $k\in \mathcal{K}$ is in the order of 
\begin{align*}
&O(|\mathcal{K}|)	\cdot O(|\mathcal{W}|)+\sum_{j\in[|\mathcal{W}|-1]}{O(|\mathcal{S}_{j}|\cdot \log |\mathcal{S}}_{j-1}|)\notag\\
=&\widetilde{O}(n+z/\varepsilon).
\end{align*}

\item \emph{Computing $\Upsilon_{4}(\omega, \ell, k)$}. We first remark that there are $O(\log^{2} (K/\varepsilon)/\varepsilon^{2})$ types of elements in $\mathcal{S}_{2}(\omega)$, here two elements are of the same type, if and only if both their weights and profits are identical to each other. This is because that the profits and weights are rounded into integer powers of $(1+\varepsilon)$, hence there are $O(\log (K/\varepsilon)/\varepsilon)$ types of profits and weights.

Now we dualize the budget constraint through a non-negative Lagrangian multiplier $\mu$. It holds that $\Upsilon_{4}(\omega, \ell,k)=\min_{\mu\geq 0}{L(\mu,\omega,\ell,k)}$, where
\begin{align*}
&L(\mu,\omega,\ell, k)\\
=&\max_{x_{e}\in [0,1]}\Big\{\sum_{e\in \mathcal{S}_{2}(\omega)}{p_{e} x_{e}}+\mu \Big(\omega-\sum_{e\in \mathcal{S}_{2}(\omega)}{w_{e}x_{e}}\Big) \Big|\sum_{e\in \mathcal{S}_{2}(\omega)}{x_{e}}\leq k-\ell \Big\}\\
=&\max_{x_{e}\in [0,1]}\Big\{\mu\omega+\sum_{e\in \mathcal{S}_{2}(\omega)}{p^{\prime}_{e}(\mu)\cdot x_{e}}\Big|\sum_{e\in \mathcal{S}_{2}(\omega)}{x_{e}}\leq k-\ell\Big\},	
\end{align*}
and profit $p^{\prime}_{e}(\mu)=p_{e}-\mu \omega_{e}$. For any fixed value of $\mu$, $\omega$ and $k$, function $L(\mu,\omega,\ell,k)$ can be computed by first sorting elements in $\mathcal{S}_{2}(\omega)$ in non-increasing order of $p^{\prime}_{e}(\mu)$, then selecting the top $k-\ell$ elements with non-negative value of $p^{\prime}_{e}(\mu)$. This can be done within $\widetilde{O}(\log^{2} (K/\varepsilon)/\varepsilon^{2})$ time.

Note that $L(\mu,\omega,\ell,k)$ is convex with respect to $\mu$, as it is the point-wise supremum of a family of linear functions in $\mu$. In particular, as long as the order of the elements remain unchanged, $L(\mu,\omega,\ell,k)$ is a linear function with respect to $\mu$, with slope equal to $(\omega-\sum_{e\in \mathcal{S}_{2}(\omega)}w_{e}x_{e})$. Hence $L(\mu,\omega,k)$ is a piecewise linear function of $\mu$. As a consequence, the optimal multiplier $\mu^{*}$ must belong to set 
\begin{align*}
\mathcal{B}=\Big\{\mu\Big|\mbox{ there exist } e^{(1)},e^{(2)}\in \mathcal{S}_{2}(\omega) \mbox{ such that } p_{e^{(1)}}(\mu)=p_{e^{(2)}}(\mu)\Big\},
\end{align*}
which can be formally represented as
\begin{align}
\mathcal{B}=\Big\{\frac{p_{e^{(1)}}-p_{e^{(2)}}}{w_{e^{(1)}}-w_{e^{(2)}}}\Big|e^{(1)},e^{(2)}\in \mathcal{S}_{2}(\omega) \Big\}\subseteq \Big\{\frac{OPT}{\omega}\cdot b \Big|b\in \mathcal{B}^{\prime} \Big\},	
\end{align}
where 
\begin{align}
\mathcal{B}^{\prime}=\Big\{(1+\varepsilon)^{b}\cdot \frac{(1+\varepsilon)^{c}-1}{(1+\varepsilon)^{d}-1}\Big||b|,|c|,|d|\leq \log(K/\varepsilon)/\varepsilon, \mbox{ and }b,c,d\in \mathbb{Z}\Big\}.
\end{align}
This follows from the facts that $p_{e^{(i)}}=\frac{\varepsilon \mathrm{OPT}}{K}\cdot (1+\varepsilon)^{b_{i}}$ and $w_{e^{(i)}}=\frac{\varepsilon \omega}{K}\cdot (1+\varepsilon)^{c_{i}}\;(i=1,2)$ for some integers $b_{i},c_{i} \in [\log(K/\varepsilon)/\varepsilon]$. Therefore
\begin{align*}
|\mathcal{B}|\leq |\mathcal{B}^{\prime}|=O(\log^{3}(K/\varepsilon)/\varepsilon^{3})=\widetilde{O}(1/\varepsilon^{3}).
\end{align*}

Utilizing the convexity of $L(\mu,\omega,\ell,k)$, for each fixed value of $\omega$, $\ell$ and $k$, $\Upsilon_{4}(\omega, \ell,k)$ can be computed in $O(\log |\mathcal{B}| \cdot \log^{2} (K/\varepsilon)/\varepsilon^{2})=\widetilde{O}(1/\varepsilon^{2})$ time, by figuring out $\mu^{*}$ via binary search over set $\mathcal{B}^{\prime}$. It is worth pointing out that $\mathcal{B}^{\prime}$ must be computed and sorted in advance, which takes $O(|\mathcal{B}^{\prime}|\log |\mathcal{B}^{\prime}|)=\widetilde{O}(1/\varepsilon^{3})$ time.
	
\end{itemize}
To summarize, our second type of relaxation $\{\Upsilon_{2}(\omega, k)\}_{\omega\in \mathcal{W}, k\in [K]}$ can be obtained in 
\begin{align}
&\underbrace{\widetilde{O}(1/\varepsilon^{3})}_{\mbox{ $\mathcal{B}^{\prime}$}}+\underbrace{O(|\mathcal{K}|\cdot|\mathcal{W}|/\varepsilon^{2})}_{\mbox{$\Upsilon_{4}(\omega, \ell, k)$ } }+\underbrace{\widetilde{O}(n+1/\varepsilon^{2})}_{\mbox{$\Upsilon_{3}(\omega, \ell)$ } } \notag\\
=& \widetilde{O}(n)+O(z/\varepsilon^{3})
\end{align}
time, and requires $O(n+|\mathcal{K}|\cdot |\mathcal{W}|)=O(n+z/\varepsilon)$ space.

\section{Putting The Pieces Together--Combining Small and Large Items}\label{mainalgosec}
In our main algorithm, we utilize our two algorithms established in Section~\ref{largeitemsubsection} and~\ref{smallputting} as two basic building blocks, to approximately enumerate all the possible profit allocations among $\mathcal{L}$ and $\mathcal{S}$. The details are specified in Appendix \ref{appendixmainalgorithm} and performance guarantee is given by Theorem~\ref{mainalgoguaran}. We remark that set $X^{\prime}$ in the algorithm is not equal to $X$ but a subset of $X$, and is given by $X^{\prime}=\{i\cdot \varepsilon \mathrm{OPT}|i\in [1/\varepsilon]\}$. 
	
\begin{theorem}\label{mainalgoguaran} The total profits of items in set $S_{o}$ returned by Algorithm \ref{mainalgorithm}, is no less than $(1-\varepsilon)\cdot \mathrm{OPT}$. Algorithm \ref{mainalgorithm} requires $O(n+z^{2}/\varepsilon)$ space and a running time of $\widetilde{O}(n+z^{4}+(z^{2}/\varepsilon)\cdot\min\{n,\varepsilon^{-1}\})=\widetilde{O}(n+z^{2}/\varepsilon^{2})$.
\end{theorem}
\begin{proof}
See Appendix \ref{appendixmaintheo}.
\end{proof}

Recall that our ultimate objective is to retrieve the solution set that has almost optimal objective function value. For large items, it can be obtained from the convolution algorithm by keeping track of the optimal allocation of the budget and cardinality bound. As for the collection of small items, let $\mathbf{x}^{*}=\{x^{*}_{e}\}_{e\in \mathcal{S}}$ be the optimal solution to the continuous problem (\ref{linrelaxation}) or (\ref{secondrelax}), we use the corresponding integer components $\{e|x^{*}_{e}=1\}$ as the approximate solution.

\section{Conclusion}\label{conclusion}
In this paper we proposed a new FPTAS for the \emph{$K$-item knapsack problem} (and \emph{Exactly $K$-item knapsack problem}) that exhibits $\widetilde{O}(K)$ and $O(z)$ improvements in time and space complexity respectively, compared with the state-of-the-art~\cite{mastrolilli2006hybrid}. More importantly, our result suggests that for a fixed value of $\varepsilon$, an $(1-\varepsilon)$-approximation solution of $K$KP can be computed in time asymptotically independent of cardinality bound $K$. Our scheme is also the first FPTAS that achieves better time and space complexity (up to logarithmic factors) than the standard dynamic programming scheme in~\cite{caprara2000approximation} over all parameter regimes.

\bibliography{refbib}
\bibliographystyle{plainnat}	

\appendix

\section{Supplementary Preliminaries}\label{appendixpre}
	
\subsection{Exact $K$-item Knapsack Problem}\label{appendixekkp}
The \emph{Exact $K$-item Knapsack Problem} (E-$K$KP) is another variant of the knapsack problem which has a deep connection with $K$KP, and can be formally formulated via replacing the cardinality upper bound constraint by an equality constraint $\sum_{i\in E}{x_{i}}=K$. It has been shown in~\citep{caprara2000approximation} that E-$K$KP and $K$KP can be converted into each other, \ie, any instance of one problem can be solved by using the algorithm of the other problem. We claim that our results presented in this paper work for E-$K$KP as well, which is straightforward to verify.
	
\subsection{Knowledge of the value of $\mathrm{OPT}$}\label{appendixknowopt} Notice that an $1/2$-approximate solution could be obtained in $O(n)$ time by properly rounding the real-valued solution of its linear programming relaxation to a feasible solution set~\citep{caprara2000approximation}. Hence, in this paper, for clarity of presentation, we assume that we know the value of $\mathrm{OPT}$. Indeed it can be verified that, if we replace $\mathrm{OPT}$ by $2\mathrm{OPT}^{\prime}$, where $\mathrm{OPT}^{\prime}$ denotes the objective value of the $1/2$-approximate solution, all of the analyses in this paper will still hold.  
	
\subsection{Upper bound on the number of non-empty classes}\label{nonemptyclass}
The correctness of bound (\ref{rupperbound}) is straightforward. Based on the definition of $r_{\mathcal{L}}, r_{\mathcal{S}}$ in Definition~\ref{classpartition}, it can be seen that $(1+\varepsilon)^{\max\{r_{\mathcal{L}},r_{\mathcal{S}}\}}\leq K/\varepsilon$ holds, combining this with the fact that there are at most $n$ non-empty classes, we conclude that (\ref{rupperbound}) is true. 
		
\section{Additional Related Work}
Here we give a brief overview of the relevant work with respect to the classic variant named \emph{Unbounded Knapsack Problem} (UKP)~\citep{ibarra1975fast}, in which the number of copies of each item could be any non-negative integer, instead of being restricted to a boolean variable as in the 0-1 KP. The earliest FPTAS for UKP was due to~\citep{ibarra1975fast}, which is an extension of their algorithm for 0-1 KP. The scheme achieves a time complexity of $O(n+1/\varepsilon^{4}\log(1/\varepsilon))$ and space complexity of $O(n+1/\varepsilon^{3})$. A more efficient  FPTAS was designed by~\citep{lawler1977fast}, which runs in $O(n+1/\varepsilon^{3})$ and requires $O(n+1/\varepsilon^{2})$ space. Recently an $\widetilde{O}(1/\varepsilon)$ improvement on both time and space complexity was made by~\citep{jansen2018faster}, in which a new FPTAS was presented with running time of $O(n+1/\varepsilon^{2}\log^{3}(1/\varepsilon))$ and $O(n+1/\varepsilon \log^{2}(1/\varepsilon))$ space bound.

\section{Proof of Proposition~\ref{upperboundsmall}}\label{appendixpro3.1}
\begin{proof}
The first result is due to the simple fact that in each class $\mathcal{S}^{\dag}_{i}$, we can retain the $K$ most profitable items, \ie, items with the smallest weights, and eliminate the other ones. Hence we have $|\mathcal{S}|\leq \min\{Kr,n\}$. On the other hand, notice that $|O^{*}\cap \mathcal{L}|\leq \mathrm{OPT}/\min_{e\in O^{*}\cap \mathcal{L}}{p_{e}}\leq\varepsilon^{-1}$, together with the fact that $|O^{*}\cap \mathcal{L}|\leq |O^{*}|\leq K$, we know that Proposition \ref{upperboundsmall} follows.
	\end{proof}

\section{Supplementary Materials of Section \ref{largeitemsubsection}}

\subsection{Proof of Lemma~\ref{continuityvarphi}}\label{appendixlemten}
	\begin{proof}
An important observation in the proof is, there are $O(z)$ number of effective convolutions in total, as the number of large items is always no more than $z$. This enables us to relate the convergence rate with the number of large items, instead of the number of classes in $\{\mathcal{L}_{i}\}_{i\in [r_{\mathcal{L}}]}$. In the following, we formalize our intuition and present a rigorous proof.
	 

	Let $O^{*}_{\omega,k}$ denote the optimal solution to subproblem for large items and $\mathcal{L}_{\omega,k}^{*(i)}=\mathcal{L}^{(i)}\cap O_{\omega,k}^{*}$ be the optimal elements in $\mathcal{L}^{(i)}$. For notational convenience, we let $x^{*(i)}_{\omega, k}=\sup \{x\in X| x\leq p(\mathcal{L}^{*(i)}_{\omega, k}) \}$, and $x^{*}_{\omega,k}=\sup\{x\in X|x\leq \sum_{i=1}^{\ell}{x^{*(i)}_{\omega, k}}\}$. 
	
	Observe that for set $\mathcal{L}^{(i)}$, 
	\begin{align}
	\phi^{X}_{\mathcal{L}^{(i)}}(x^{*(i)}_{\omega, k},|\mathcal{L}^{*(i)}_{\omega, k}|)=\phi_{\mathcal{L}^{(i)}}(x^{*(i)}_{\omega, k},|\mathcal{L}^{*(i)}_{\omega, k}|)\leq \phi_{\mathcal{L}^{(i)}}(p(\mathcal{L}^{*(i)}_{\omega, k}),|\mathcal{L}^{*(i)}_{\omega, k}|)=w(\mathcal{L}^{*(i)}_{\omega, k}), \label{phiibound}
	\end{align}
	where the first equality holds since $\phi^{X}_{\mathcal{L}^{(i)}}$ is the restriction of $\phi_{\mathcal{L}^{(i)}}$ to $X$, and the inequality follows from the fact that $\phi_{(\cdot)}(\cdot,\cdot)$ is monotone non-decreasing with respect to profit $p$. Combining (\ref{phiibound}) with the subadditivity of
	inverse weight function,
	\begin{align}
	\label{boundphix}
	(\otimes_{i=1}^{\ell}{\phi^{X}_{\mathcal{L}^{(i)}}})(x^{*}_{\omega,k},k)\leq \sum_{i=1}^{\ell}{\phi^{X}_{\mathcal{L}^{(i)}}(x^{*(i)}_{\omega, k},|\mathcal{L}^{*(i)}_{\omega, k}|)} \leq \sum_{i=1}^{\ell}{w(\mathcal{L}^{*(i)}_{\omega, k})}\leq\omega,
	\end{align}
	which further implies that $\varphi^{X}_{\mathcal{L}}(\omega,k)\geq x^{*}_{\omega,k}\geq  \sum_{i=1}^{\ell}{x^{*(i)}_{\omega, k}}-\delta_{X}$.
	Hence, we are able to lower bound the error incurred by discretization as,
	\begin{align}
	\label{errordiscretization}
	\varphi^{X}_{\mathcal{L}}(\omega,k)-\varphi_{\mathcal{L}}(\omega,k)\geq &\sum_{i=1}^{\ell}{x^{*(i)}_{\omega, k}}-\sum_{i=1}^{\ell}{p(\mathcal{L}_{\omega,k}^{*(i)})}-\delta_{X} \geq -\delta_{X}\Big(1+\sum_{i=1}^{\ell}{\mathbbm{1}_{\Delta_{i}\neq 0}}\Big),
	\end{align}
	where $\Delta_{i}=\mathcal{L}_{\omega,k}^{*(i)}-x^{*(i)}_{\omega,k}$, and the second inequality holds because $\Delta_{i}\leq \delta_{X}$. To bound the RHS of (\ref{errordiscretization}), we note that $\Delta_{i}\neq 0$ only if $\mathcal{L}^{*(i)}_{\omega,k}$ is non-empty. Hence
	\begin{align}
	\label{indibound}
	\sum_{i=1}^{\ell}{\mathbbm{1}_{\Delta_{i}\neq 0}}\leq \sum_{i=1}^{\ell}{\mathbbm{1}_{\mathcal{L}^{*(i)}_{\omega,k}\neq \emptyset}}\leq |\mathcal{L}\cap O^{*}_{\omega,k}|\leq z,
	\end{align}
	and the error brought by discretization is lower bounded by $-(z+1)\delta_{X}$. On the other hand, it is clear that $\phi^{X}_{\mathcal{L}}(\omega, k)\leq \phi_{\mathcal{L}}(\omega, k)$, which follows by applying induction on $|\mathcal{L}|$. Therefore the absolute value of the error is no more than $(z+1)\delta_{X}$. Combining this with the fact that $\delta_{X}\geq \frac{\mathrm{OPT}}{|X|}$, the proof is complete.
	\end{proof}

\subsection{Proof of Lemma~\ref{monoineq}}\label{appendixlem12}
\begin{proof} Let $\Delta=[\chi_{H}(\zeta_{2})-\chi_{H}(\zeta_{1})] - [\zeta_{2}-\zeta_{1}]$ denote the difference between the numerator and denominator in inequality~\ref{gradientineq}. We assume that $\Delta>0$ and finish the proof by contradiction.

Without loss of generality we can assume that $\zeta_{2}\geq \zeta_{1}$. We first consider points $(p_{0}+\lambda_{a}\zeta_{1},k_{0}+\zeta_{1},\chi_{H}(\zeta_{1}))$ and $(p_{0}+\lambda_{a}\zeta_{1},k_{0}+\zeta_{1},\chi_{H}(\zeta_{1})+\Delta)$ in column indexed by $\zeta_{1}$, the following inequality follows from the fact that $\chi_{H}(\zeta_{1})$ is the index of column minimum,
\begin{align}
&\phi^{X}_{\mathcal{L}^{\dag}_{a}}(\lambda_{a}\chi_{H}(\zeta_{1}), \chi_{H}(\zeta_{1}))+\phi^{X}_{S}(p_{0}+\lambda_{a}[\zeta_{1}-\chi_{H}(\zeta_{1})], k_{0}+[\zeta_{1}-\chi_{H}(\zeta_{1})])\label{ineq11}\\
\leq &\phi^{X}_{S}(p_{0}+\lambda_{a}[\zeta_{1}-(\chi_{H}(\zeta_{1})+\Delta)], k_{0}+[\zeta_{1}-(\chi_{H}(\zeta_{1})+\Delta)])\notag\notag\\
&+\phi^{X}_{\mathcal{L}^{\dag}_{a}}(\lambda_{a}[\chi_{H}(\zeta_{1})+\Delta], \chi_{H}(\zeta_{1})+\Delta)
\label{ineq12}.
\end{align}
Similar results can be obtained for points $(p_{0}+\lambda_{a}\zeta_{2},k_{0}+\zeta_{2},\chi_{H}(\zeta_{2}))$ and $(p_{0}+\lambda_{a}\zeta_{2},k_{0}+\zeta_{2},\chi_{H}(\zeta_{2})-\Delta)$ in the cube,
\begin{align}
&\phi^{X}_{\mathcal{L}^{\dag}_{a}}(\lambda_{a}\chi_{H}(\zeta_{2}), \chi_{H}(\zeta_{2}))+\phi^{X}_{S}(p_{0}+\lambda_{a}[\zeta_{2}-\chi_{H}(\zeta_{2})], k_{0}+[\zeta_{2}-\chi_{H}(\zeta_{1})])\label{ineq21}\\
\leq &\phi^{X}_{S}(p_{0}+\lambda_{a}[\zeta_{2}-(\chi_{H}(\zeta_{2})-\Delta)], k_{0}+[\zeta_{2}-(\chi_{H}(\zeta_{2})-\Delta)])\notag\\
&+\phi^{X}_{\mathcal{L}^{\dag}_{a}}(\lambda_{a}[\chi_{H}(\zeta_{2})-\Delta], \chi_{H}(\zeta_{2})-\Delta).\label{ineq22}
\end{align}
We remark that $\chi_{H}(\zeta_{1})+\Delta$ is a valid index as $\chi_{H}(\zeta_{1})<\chi_{H}(\zeta_{1})+\Delta=\chi_{H}(\zeta_{2})-[\zeta_{2}-\zeta_{1}]\leq \chi_{H}(\zeta_{2})$. Similar arguments can be applied to show the validity of index $\chi_{H}(\zeta_{2})-\Delta$. According to the definition of $\Delta$, 
\begin{align}
\label{ineq3}
p_{0}+\lambda_{a}[\zeta_{2}-(\chi_{H}(\zeta_{2})-\Delta)]=p_{0}+\lambda_{a}[\zeta_{1}-\chi_{H}(\zeta_{1})],
\end{align}
which suggests that the term in (\ref{ineq22}) is identical to that in (\ref{ineq12}). Via similar reasoning, we have 
\begin{align}
\label{ineq4}
p_{0}+\lambda_{a}[\zeta_{1}-(\chi_{H}(\zeta_{1})+\Delta)]=p_{0}+\lambda_{a}[\zeta_{2}-\chi_{H}(\zeta_{2})].
\end{align}
Substituting (\ref{ineq3})-(\ref{ineq4}) into (\ref{ineq11})-(\ref{ineq22}), it holds that 
\begin{align}
\label{conineq}
&\phi^{X}_{\mathcal{L}^{\dag}_{a}}(\lambda_{a}\chi_{H}(\zeta_{1}), \chi_{H}(\zeta_{1}))-	\phi^{X}_{\mathcal{L}^{\dag}_{a}}(\lambda_{a}[\chi_{H}(\zeta_{1})+\Delta], \chi_{H}(\zeta_{1})+\Delta)\notag\\
\leq &\phi^{X}_{\mathcal{L}^{\dag}_{a}}(\lambda_{a}[\chi_{H}(\zeta_{2})-\Delta], \chi_{H}(\zeta_{2})-\Delta)-\phi^{X}_{\mathcal{L}^{\dag}_{a}}(\lambda_{a}\chi_{H}(\zeta_{1}), \chi_{H}(\zeta_{1})).
	\end{align}
	Recall that $\phi^{X}_{\mathcal{L}^{\dag}_{a}}(\lambda_{a}t, t)=\min\{{\sum_{e\in S}}{w_{e}}|S\subseteq \mathcal{L}^{\dag}_{a}, \sum_{e\in S}{p_{e}}\geq \lambda_{a}t\}$, in which the cardinality upper bound is redundant, as the profit of each single item is no less than $\varepsilon \mathrm{OPT}$. Without loss of generality we can assume that $w_{i} \neq w_{j}$, otherwise we can slightly change $w_{i}$ and the budget to achieve this goal. Therefore $\phi^{X}_{\mathcal{L}^{\dag}_{a}}(\lambda_{a}t, t)$ is a strictly convex function and (\ref{conineq}) does not hold. The proof is complete.
	\end{proof} 	

\subsection{Details of Algorithm \ref{sliceindexalgo}}\label{appendixsliceindex}

\begin{algorithm}[H]
\small
    \caption{SliceIndex$(H)$}
    \label{sliceindexalgo}
\textbf{Input:} $H$;\\
\textbf{Output:} $\chi_{H}(\cdot)$\\
$H^{\prime}\leftarrow$ Even fibers in $H$;\\
Compute the optimum indices of fibers in $H^{\prime}$ via SliceIndex$(H^{\prime})$;\\
\For{each odd fiber $2i$ in $H$}
{Enumerate $[\chi_{H}(2i+1)-1, \chi_{H}(2i-1)+1]$ to find the minimum index in the $2i$-th fiber of $H$.}
\end{algorithm}

\subsection{Proof of Proposition~\ref{sliceddc}}\label{appendixsliceddc}
\begin{proof}
We use $c_{H}$ to denote the number of columns in slice $H$ and let $\mathcal{T}_{H}(c_{H})$ be the running time of computing the index function for $H$. Then 
\begin{itemize}
\item Line $4$ requires $\mathcal{T}_{H}(c_{H}/2)$	 time. Without loss of generality we can assume that $c_{H}$ is even, otherwise it can be verified that the corresponding total time complexity is within the same order.
\item Each iteration in line $6$ takes 
\begin{align}\label{oddcompl}
&O([\chi_{H}(2i-1)+1]-[\chi_{H}(2i+1)-1]+1)\notag\\
=&O(\chi_{H}(2i-1)-\chi_{H}(2i+1)+3)
\end{align}
time. We remark that the RHS of (\ref{oddcompl}) is non-negative according to Lemma \ref{monoineq}. Taken together, the running time of computing column minimum in odd columns can be upper bounded as,
\begin{align}\label{odd}
&\sum_{i=1}^{O(c_{H}/2)}{O(\chi_{H}(2i-1)-\chi_{H}(2i+1)+3)}=O(c_{H}+z)=O(z), 
\end{align}
which holds because both $c_{H}$ and $\chi_{H}(\zeta)$ are no more than $z$. 
\end{itemize}
To summarize, the total running time satisfies the recurrence relation $\mathcal{T}_{H}(c_{H})=\mathcal{T}_{H}(c_{H}/2)+O(z)$. Solving this equation we have $\mathcal{T}_{H}(c_{H})=O(z\log z)$, the proof is complete.
\end{proof}

\subsection{Details of Algorithm \ref{convolution}}\label{appendixconvolution}
\begin{algorithm}[H]
\small
    \caption{Convolution Algorithm $\otimes$}
    \label{convolution}
\textbf{Input:} $\phi^{X}_{\mathcal{L}^{\dag}_{a}}(\cdot,\cdot)$, $\phi^{X}_{S}(\cdot,\cdot)$;\\
\textbf{Output:} $(\phi^{X}_{\mathcal{L}^{\dag}_{a}}\otimes \phi^{X}_{S})(\cdot,\cdot)$\\
\For{each slice $H$ in the form of (\ref{slicedef}) }
{Compute $\chi_{H}(\cdot)$ using SliceIndex($H$);\\}
\For {$p \in X$, $k\in [K]$}{
$(\phi^{X}_{\mathcal{L}^{\dag}_{a}}\otimes \phi^{X}_{S})(p,k)=\phi^{X}_{\mathcal{L}^{\dag}_{a}}(\max\{x\in X: x\leq \psi(p,k)\cdot p^{\dag}_{a}\}, \psi(p,k))+\phi^{X}_{S}(\max\{x\in X: x\leq p-\psi(p,k)\cdot p^{\dag}_{a}\}, k-\psi(p,k))\}\;(p\in X)$\\
}
\textbf{Return $(\phi^{X}_{\mathcal{L}^{\dag}_{a}}\otimes \phi^{X}_{S})(\cdot,\cdot)$}.
\end{algorithm}

\subsection{Proof of Lemma~\ref{largetimecomplexity}}~\label{appendixle}
Based on Proposition \ref{sliceddc}, it can be seen that a single convolution operation takes $O(|X|\cdot z\log z)=\widetilde{O}(z^{2}/\varepsilon)$ time, since there are $O(|X|)$ slices in the searching space. Additionally, in Algorithm \ref{highlevel} we need to
\begin{itemize}
\item Compute the base functions $\phi_{\mathcal{L}^{\dag}_{i}}(\cdot,\cdot)\; (i\in [r_{\mathcal{L}}])$, which requires $O(n)$ time. This is achieved by storing sequence
\begin{align*}
\min\Big\{\sum_{e\in T}{}w_{e}\Big|T\subseteq \mathcal{L}^{\dag}_{i}, |T|=j\Big\}_{i\in [r], j\in [|\mathcal{L}^{\dag}_{i}|]}    
\end{align*}
and then utilizing binary search on the sequence, which requires $O(\sum_{i=1}^{r}{|\mathcal{L}^{\dag}_{i}|})=O(n)$ space and time in advance.
\item Perform $\otimes$ operation for $r=O(\min\{\log(1/\varepsilon)/\varepsilon,n\})$ times, the time complexity of which is $O((z^{2}r\log z)/\varepsilon)=O((z^{2}\log z/\varepsilon)\cdot \min \{ \log(1/\varepsilon)/\varepsilon, n\})$.
\end{itemize}
The proof is complete. 
	

\section{Supplementary Materials of Section \ref{smallputting} }	
\subsection{Proof of Lemma~\ref{continuousrelaxationmainlemma}}\label{appendixcontinuousrelaxationmainlemma}
\begin{proof}
The complexity results in Lemma~\ref{continuousrelaxationmainlemma} can be achieved by letting $\varphi^{\dag}_{\mathcal{S}}=\Upsilon_{1}$ when $K\leq \varepsilon^{-1}$ and $\varphi^{\dag}_{\mathcal{S}}=\Upsilon_{2}$ otherwise. Therefore, the total running time is bounded by 
\begin{align}
&O\Big(\frac{z}{\varepsilon}\cdot\min\Big\{\frac{K}{\varepsilon}, n\Big\}\Big)\cdot\mathbbm{1}_{\{K\leq \varepsilon^{-1}\}}+	\Big[\widetilde{O}\Big(\min\Big\{\frac{K}{\varepsilon^{2}},\frac{n}{\varepsilon}\Big\}\Big)+O\Big(\frac{z}{\varepsilon^{3}}\Big)\Big]\cdot \mathbbm{1}_{\{K> \varepsilon^{-1}\}}\notag\\
=&\widetilde{O}\Big( \min\Big\{\frac{z^{2}}{\varepsilon^{2}}, \frac{nz}{\varepsilon}\Big\}+z^{4}+\min\Big\{Kz^{2},nz\Big\}\Big).
\end{align}
The space required is in the order of $O(|\mathcal{K}||\mathcal{W}|)=O(z/\varepsilon)$.
\end{proof}

\subsection{Proof of Lemma~\ref{lemmaapproxi}}\label{appendixlemmaapproxi}		
\begin{proof}
It suffices to show the following bounds on the differences between functions $\varphi_{\mathcal{S}}$ and $\Upsilon_{1}$, $\Upsilon_{2}$:
\begin{align}
&|\Upsilon_{1}(\omega,k)	-\varphi_{\mathcal{S}}(\omega,k)|\leq 2\varepsilon \mathrm{OPT}, \label{approximatebound0}\\
&|\Upsilon_{2}(\omega,k)-\varphi_{\mathcal{S}}(\omega,k)|\leq 4\varepsilon \mathrm{OPT}\label{approximatebound}.
\end{align}
Assuming inequalities above, we can complete the proof. Now we proceed to prove the bounds (\ref{approximatebound0}) and (\ref{approximatebound}). 

\paragraph{(\uppercase\expandafter{\romannumeral 1}) Proof of bound (\ref{approximatebound0})} We first make the observation that $\varphi_{\mathcal{S}}(\omega,k)\leq\Upsilon_{1}(\omega,k)$, since the feasible region in $\varphi_{\mathcal{S}}$ is a subset of that in $\Upsilon_{1}$. As it has been shown in~\citep{caprara2000approximation}, there are at most two fractional components in $\mathbf{x}^{*}$, the optimal solution to the LP relaxation (\ref{linrelaxation}). Hence, the objective value will suffer a loss of at most $2\varepsilon \mathrm{OPT}$, if we set all the fractional entries in $\mathbf{x}^{*}$ to be $0$, \ie, 
\begin{align*}
\sum_{i=1}^{n}{p_{i}\bar{x}_{i}}\geq \Upsilon_{1}(\omega,k)-2\varepsilon \mathrm{OPT}
\end{align*}
holds for the integer vector $\bar{\mathbf{x}}^{*}$. On the other hand, notice that $\bar{\mathbf{x}}^{*}$ is also a feasible solution to the subproblem $\varphi_{\mathcal{S}}(\omega, k)$, it follows that $\sum_{i=1}^{n}{p_{i}\bar{x}_{i}}\leq \varphi_{\mathcal{S}}(\omega,k)$. Relating $\varphi_{\mathcal{S}}(\omega,k)$ and $\Upsilon_{1}(\omega,k)$ to the total profits of $\bar{\mathbf{x}}$, (\ref{approximatebound0}) follows. 
	
\paragraph{(\uppercase\expandafter{\romannumeral 2}) Proof of bound (\ref{approximatebound})}	To show the correctness of (\ref{approximatebound}), observe that each one of the following two operations appearing in the definition of $\Upsilon_{4}(\omega,\ell,k)$, will incur a multiplicative loss of at most $(1-\varepsilon)$, compared with the LP relaxation on set $\mathcal{S}_{2}(\omega)$, denoted by $\Upsilon_{1}(\omega,k,\mathcal{S}_{2}(\omega))$:
\begin{itemize}
\item Increasing the weight $w_{e}\;(e\in\mathcal{S}_{2}(\omega))$ to $\bar{w}_{e} \in [w_{e}, (1+\varepsilon) w_{e}]$;
\item Scaling the budget $\omega$ by a factor of $(1-\varepsilon)$. 
\end{itemize}
Therefore $\Upsilon_{4}$ can be lower bounded using  $\Upsilon_{1}(\omega,k,\mathcal{S}_{2}(\omega))$:
\begin{align}\label{varphi12differ}
\Upsilon_{4}(\omega,k)\geq (1-\varepsilon)^2\cdot\Upsilon_{1}(\omega,k,\mathcal{S}_{2}(\omega))\overset{}{\geq} \varphi_{\mathcal{S}_{\omega}}(\omega,k)-4\varepsilon \mathrm{OPT}.
\end{align}
The last inequality $(a)$ follows from the fact that $(1-\varepsilon)^{2}\geq 1-2\varepsilon$, together with inequality $\Upsilon_{1}(\omega,k,\mathcal{S}_{2}(\omega))\geq \varphi_{\mathcal{S}_{\omega}}(\omega,t)-2\varepsilon \mathrm{OPT}$, whose proof goes along the same lines as the proof of (\ref{approximatebound0}). 

Let $\mathcal{S}^{*}(\omega,k)$ be the optimal solution set to $\varphi_{\mathcal{S}}(\omega,k)$. Observe that the profit function $\varphi_{\mathcal{S}}$ can be expressed as 
\begin{align}\label{ineqappendix25}
\varphi_{\mathcal{S}}(\omega,k)=\varphi_{\mathcal{S}}(\omega-w(\mathcal{S}^{*}_{1}(\omega,k)),k-|\mathcal{S}^{*}_{1}(\omega,k)|)+p(\mathcal{S}^{*}_{1}(\omega,k))	,
\end{align}
where $\mathcal{S}^{*}_{1}(\omega,k)=\mathcal{S}^{*}(\omega,k)\cap\{e\in \mathcal{S}|w_{e}\leq \varepsilon \omega/K\}$ represents the set of elements in $\mathcal{S}^{*}(\omega,k)$ with cost no more than $\varepsilon \omega/ K$. As a consequence, the difference between $\Upsilon_{2}$ and $\varphi_{\mathcal{S}}$ can be lower bounded as,
\begin{align*}
&\Upsilon_{2}(\omega,k)-\varphi_{\mathcal{S}}(\omega,k)\notag\\
\overset{}{\geq} & [\Upsilon_{4}(\omega-\mathcal{S}^{*}_{1}(\omega,k), k-|\mathcal{S}^{*}_{1}(\omega,k)|)-\varphi_{\mathcal{S}}(\omega-w(\mathcal{S}^{*}_{1}(\omega,k)),k-|\mathcal{S}^{*}_{1}(\omega,k)|)]\\
&+[\Upsilon_{3}(|\mathcal{S}^{*}_{1}(\omega,k)|)-p(\mathcal{S}^{*}_{1}(\omega,k))]\notag\\
\overset{}{\geq} & -4\varepsilon \mathrm{OPT},
\end{align*}
where the last inequality follows from (\ref{varphi12differ}) and the fact that $\Upsilon_{3}(|\mathcal{S}^{*}_{1}(\omega,k)|)\geq p(\mathcal{S}^{*}_{1}(\omega,k))$.

Finally we conclude that $\Upsilon_{4}(\omega,k)\leq \varphi_{\mathcal{S}}(\omega,k)+ 2\varepsilon \mathrm{OPT}$. Let $\ell^{*}_{\omega, k}$ be the optimal index in $\Upsilon_{4}(\omega,k)$, we consider set $\tilde{\mathcal{S}}$ consists of the following two types of items: 
\begin{itemize}
\item Top $\ell^{*}_{\omega,k}$ elements in $\mathcal{S}_{1}(\omega)$;
\item Elements corresponding to the integer entries in the optimal solution to $\Upsilon_{4}(\omega,k-\ell^{*}_{\omega,k})$.
\end{itemize}
Observe that the indicator vector of $\tilde{\mathcal{S}}$ is a feasible solution to $\varphi_{\mathcal{S}}(\omega,k)$, hence $p(\tilde{\mathcal{S}})\leq \varphi_{\mathcal{S}}(\omega,k)$. Combining with the fact that $p(\tilde{\mathcal{S}})\geq \Upsilon_{4}(\omega,k)-2\varepsilon \mathrm{OPT}$, the proof is complete. 
\end{proof}

\section{Supplementary Materials of Section \ref{mainalgosec} }	

\subsection{Details of Algorithm \ref{mainalgorithm}}\label{appendixmainalgorithm}

\begin{algorithm}[H]\label{subroutine2}
\small
    \caption{Main Algorithm}
    \label{mainalgorithm}
\textbf{Input:} Functions $\phi^{X}_{\mathcal{L}}(\cdot,\cdot)$, $\tilde{\varphi}_{\mathcal{S}}(\cdot,\cdot)$;\\
\textbf{Output:} Near optimal solution $T_{o}$;\\
$(k^{*},x^{*})\leftarrow \argmax_{k\in [z], x\in X^{\prime}}\Big\{x+\varphi^{\dag}_{\mathcal{S}}(W-\phi^{X}_{\mathcal{L}}(x,k),K-k)\Big\}$;\\
$T^{\mathcal{S}}_{o} \leftarrow $ The solution set in $\mathcal{S}$ corresponding to $\varphi^{\dag}_{\mathcal{S}}(W-\phi_{\mathcal{L}}(x^{*},k^{*}),K-k^{*})$;\\
$T^{\mathcal{L}}_{o} \leftarrow $The solution set in $\mathcal{L}$ corresponding to $\phi_{\mathcal{L}}(x^{*}, k^{*})$; \\
\textbf{Return} $S_{o} \leftarrow T^{\mathcal{S}}_{o} \cup T^{\mathcal{L}}_{o} $. 
\end{algorithm}

\subsection{Proof of Theorem~\ref{mainalgoguaran}}\label{appendixmaintheo}
\begin{proof}
Without loss of generality we can assume that 
\begin{align}
\phi^{X}_{\mathcal{L}}(\mathrm{OPT}-\delta_{X^{\prime}},|O^{*}\cap \mathcal{L}|)> w(O^{*}\cap \mathcal{L}).
\end{align}
Otherwise, the optimal solution in $\mathcal{S}$ already achieves a near optimal approximation. In the following, we let $x^{(*)}$ be the best approximation of $\sum_{e\in O^{*}\cap \mathcal{L}}{p_{e}}$ in $X^{\prime}$, \ie, $x^{(*)}\in X^{\prime}$ and  
\begin{align}\label{combinassump}
\phi^{X}_{\mathcal{L}}(x^{(*)},k^{(*)})\leq w(O^{*}\cap \mathcal{L})\leq \phi^{X}_{\mathcal{L}}(x^{(*)}+\delta_{X^{\prime}},k^{(*)}),
\end{align}
where $k^{(*)}=|O^{*}\cap \mathcal{L}|$. Notice that $\phi^{X}_{\mathcal{L}}$ is non-decreasing with respect to profit, we can conclude that such an $x^{(*)}$ exists. In addition, $x^{(*)}+\delta_{X^{\prime}}\in X^{\prime}$. On the other hand, we have
\begin{align}
x^{(*)}+\delta_{X^{\prime}}&\overset{(a)}{\geq} \varphi^{X}_{\mathcal{L}}(\phi^{X}_{\mathcal{L}}(x^{(*)}+\delta_{X},k^{(*)}), k^{(*)})\\ &\overset{(b)}{\geq} \varphi^{X}_{\mathcal{L}}(w(O^{*}\cap \mathcal{L}), k^{(*)})\notag\\
&\overset{(c)}{\geq} \varphi_{\mathcal{L}}(w(O^{*}\cap \mathcal{L}), k^{(*)})-\varepsilon \mathrm{OPT}, \label{lowerboundxstar}
\end{align}
where $(a)$ follows from the definition of $\phi_{\mathcal{L}}$ and $\varphi_{\mathcal{L}}$; $(b)$ is based on the monotonicity of $\varphi^{X}_{\mathcal{L}}(\cdot , |\mathcal{L}\cap O^{*}|)$ and RHS of (\ref{combinassump}); In $(c)$ we utilize the point-wise convergence property of  $\varphi^{X}$ claimed in Lemma \ref{continuityvarphi}.
	
To summarize, the total profits of $S_{o}$ can be lower bounded as,
\begin{align}
p(S_{o})=&p(S^{\mathcal{L}}_{o})+p(S^{\mathcal{S}}_{o})\notag\\
\overset{(a)}{\geq}&[\varphi^{\dag}_{\mathcal{S}}(W-\phi^{X}_{\mathcal{L}}(x^{(*)},k^{(*)}), K-k^{(*)})-4\varepsilon\cdot\mathrm{OPT}]+ x^{(*)}\notag\\
\overset{(b)}{\geq} & [\varphi^{\dag}_{\mathcal{S}}(w(\mathcal{S}\cap O^{*}),  K-k^{(*)})-4\varepsilon\cdot  \mathrm{OPT}]+[\varphi_{\mathcal{L}}(w(O^{*}\cap \mathcal{L}),k^{(*)})-\delta_{X^{\prime}}-\varepsilon\cdot  \mathrm{OPT}]\notag\\
\geq & [\varphi_{\mathcal{S}}(w(O^{*}\cap \mathcal{L}), K-k^{(*)})+\varphi_{\mathcal{L}}(w(O^{*}\cap \mathcal{L}),k^{(*)})]-6\varepsilon \cdot \mathrm{OPT} \notag\\
=&(1-6\varepsilon)\cdot \mathrm{OPT}, \notag
\end{align}
where $(a)$ comes from Lemma \ref{lemmaapproxi} and the fact that $(k^{(*)},x^{(*)})$ is a candidate pair in the $3$-th line of Algorithm \ref{mainalgorithm}. In $(b)$, the first term follows from inequality (\ref{lowerboundxstar}), the second term is due to LHS of (\ref{combinassump}) and the monotonicity of $\varphi^{\dag}_{\mathcal{S}}$. 
		
\paragraph{Complexity Results.}The time complexity result directly follows from Lemmas \ref{largetimecomplexity} and \ref{continuousrelaxationmainlemma}:
\begin{align}
&\underbrace{O(n)+\widetilde{O}\Big(\min \Big\{\frac{z^{2}}{\varepsilon^{2}}, \frac{nz^{2}}{\varepsilon} \Big\}\Big)}_{\mbox{Figure out $\{\phi^{X}_{\mathcal{L}}(x,k)\}_{k\in [z], x\in X}$}} + \underbrace{\widetilde{O}\Big(n+\min\Big\{\frac{z^{2}}{\varepsilon^{2}}, \frac{nz}{\varepsilon}\Big\}+z^{4}\Big)}_{\mbox{Compute $\varphi^{\dag}_{\mathcal{S}}$}}\\
=&\widetilde{O}\Big(n+z^{4}+\frac{z^{2}}{\varepsilon}\cdot\min\Big\{n,\varepsilon^{-1}\Big\}\Big), \label{timecomplexityspec}
\end{align} 
which is within the order of $\widetilde{O}(n+z^{2}/\varepsilon^{2})$. For the space requirement,  we need to store the information about $\phi^{X}_{\cup_{j=1}^{i-1}{\mathcal{L}^{\dag}_{j}}}$ to implement the new convolution operation in the current stage, this requires $O(|X|\cdot z)=O(z^{2}/\varepsilon)$ space. Combining with Lemmas \ref{largetimecomplexity} and \ref{continuousrelaxationmainlemma}, it can be seen that $O(n+z^{2}/\varepsilon)$ space is sufficient.
\end{proof}	

\section{Application in resource constrained scheduling}\label{Secapp}
In this section, we briefly revisit the classic resource constrained scheduling problem in~\citep{jansen2006preemptiveresource}, which asks to design a preemptive scheduling algorithm that minimizes the maximum completion time, while satisfying the resource constraint. More specifically, for a given set of tasks $\mathcal{T}=\{T_{1}, T_{2},\ldots, T_{n}\}$ and $m$ identical machines, $p_{j} \;(j\in \mathcal{T})$ units of time and $r_{j}(j\in \mathcal{T})$ units of resources are required for processing task $j$, while there are only $c$ units of resources available at each time slot. The problem is to design a scheduling algorithm to minimize $C_{max}$, the maximum completion time. As in the literature, the problem is denoted by $P|res1,\ldots,pmtn|C_{\max}$. 
	 
We remark that it is possible to obtain a faster FPTAS for this problem by following the approach in~\citep{jansen2006preemptiveresource}, which is mainly based on the linear programming formulation~\cite[Eq (1.1)]{jansen2006preemptiveresource}. For the case when there is only one resource constraint, the subproblem that need to be solved turns out to be the $K$-item knapsack problem studied in this paper. According to~\citep{jansen2006preemptiveresource},
the following proposition holds. 
\begin{proposition}[\citep{jansen2006preemptiveresource}]\label{schedulingcom}
An FPTAS for $K$-item knapsack problem with time complexity $\mathcal{T}(n,m,1/\varepsilon)$ implies a FPTAS for problem $P|res1,\ldots,pmtn|C_{\max}$ with time complexity	$O((\mathcal{T}(n,m,1/\varepsilon)+n\log\log(n/\varepsilon))\cdot n\log(1/\varepsilon)(1/\varepsilon^{2}+\log n))$.
\end{proposition}

Note that the complexity term in Proposition~\ref{schedulingcom} is proportional	to $\mathcal{T}(n,m,1/\varepsilon)+n\log\log(n/\varepsilon)$. In most parameter regimes, it is dominated by $\mathcal{T}(n,m,1/\varepsilon)$, the complexity of $K$-item knapsack problem. Roughly speaking, the complexity reduction achieving in $P|res1,\ldots,pmtn|C_{\max}$ is in the same order as the improvement obtained in $K$KP.

\section{Application in network caching}\label{applicationsub}

As the Internet traffic is dominated by popular contents (e.g., YouTube, Netflix videos) and the price of storage gets cheaper, recent Internet architectures such as Content-Centric Networking (CCN) suggest storing popular contents in network caches or routers, which could significantly reduce network congestion~\citep{jacobson2007content}. One problem arising is {\em how to choose files to store in a network cache to maximize the hit ratio (\ie, the probability that a requested file is stored in the cache)}. Let $f(S)$ denote the required storage size to store the chosen file set $S$ and $g(S)$ represent the cache miss probability for the chosen files. We consider the following cardinality constrained minimization problem, motivated by this file selection problem in network caching~\citep{meyer2012study,nam2017synccoding}, 
\begin{align}
\min_{S\subseteq E}F(S)&=f(S)+g(S),\label{subminpro}\\
s.t.\; |S|&=K, \label{submcard}
\end{align}
where function $F(\cdot):2^{E}\rightarrow \mathbbm{R}^{+}$ is the summation of $f(\cdot):2^{E}\rightarrow \mathbbm{R}^{+}$ and $g(\cdot):2^{E}\rightarrow \mathbbm{R}^{+}$, which is non-negative and additive. We remark that $f(\cdot)$ satisfies the marginal decreasing property, since we assume that caches compress the files to maximize the remaining storage, the compression efficiency\footnote{The compression efficiency is the ratio between the compressed size and the sum of original file sizes~\citep{nam2017synccoding}.} increases as more files are compressed together~\citep{nam2017synccoding}. $g(\cdot)$ is a modular function since the cache miss probability is simply the sum of the hit probabilities of the uncached files.

The formal definition of a submodular function is given as following.
\begin{definition}[Submodular function]
Set function $f(\cdot):2^{E} \rightarrow \mathbb{R}^{+}$ is \emph{submodular} if for all subsets $S,T\subseteq E$, inequality $f(S)+f(T)\geq f(S\cup T)+f(S\cap T)$ holds. $f(\cdot)$ is monotone non-decreasing if $f(S)\geq f(T)$ holds for $\forall T\subseteq S$.	
\end{definition}

The formulation of (\ref{subminpro})-(\ref{submcard}) might be not that interesting in the context of submodular optimization, as the non-increasing property of $g(\cdot)$ is rather artificial. In addition, the problem is closely related to the problem of minimizing the difference between submodular functions~\citep{10.5555/3020652.3020697}, and submodular cost submodular cover (SCSK) problem~\citep{iyer2013submodular} when the constraint function is additive, while we have an additional cardinality constraint.  Here we present an alternative approach that is standard to a certain degree, but is more straightforward. Our motivation is to show the potential application of E-$K$KP.

%

	
	
\subsection{A near-optimal algorithm}
We present a near-optimal algorithm in which the solution to E-$K$KP plays an important role. One important ingredient in the algorithm is the \emph{ellipsoid approximation}~\citep{goemans2009approximating} of a monotone submodular function.
	
\begin{definition}[Ellipsoid relaxation~\citep{goemans2009approximating}]
For any monotone submodular function $f(\cdot)$, we can construct a function $f^{\natural}(\cdot):2^{E}\rightarrow \mathbb{R}^{+}$ that approximates $f(\cdot)$ by a factor of $\alpha(n)=O(\sqrt{n}\log n)$, by issuing polynomial number of queries to $f(\cdot)$, \ie, $f^{\natural}(S)  \leq f(S) \leq \alpha(n)\cdot f^{\natural}(S)$.  Moreover, there exist $c_{e}>0$ such that $f^{\natural}(S)=\sqrt{\sum_{e\in S}{c_{e}}}$.
\end{definition}	
Owing to the simple form of $f^{\natural}$, we can reduce the problem to E-$K$KP.
	
\paragraph{Reduction to E-$K$KP.} Note that $g(\cdot)$ can represented as a constant minus a monotone non-decreasing modular function, \ie, there exists a constant $C$ and $\bar{g}(\cdot): 2^{E}\rightarrow \mathbb{R}^{+}$, such that $g(\cdot)=C-\bar{g}(\cdot)$. Consider the following problem with budget $\omega$ and cardinality bound $K$,
\begin{align}
&\max{\bar{g}(S)}\label{ekkpequiv1}\\
s.t.\; & {f^{\natural}}^{2}(S)  \leq \omega \label{softcon}\\
&|S|=K\label{ekkpequiv3}
\end{align}
where constraint (\ref{softcon}) is soft. It is equivalent to E-$K$KP because $f^{\natural}(S)(\cdot)$ is the square root of an additive function. We approximately solve problem (\ref{ekkpequiv1}) for every $\omega\in \{L^{\natural}(1+\varepsilon)^{i}|i\geq 0\}\cap [L^{\natural}, U^{\natural}]$, where $L^{\natural}=\min\{\sum_{e\in S}g(e)||S|=K\}$, $U^{\natural}=\max\{\sum_{e\in S}g(e)||S|=K\}$, then we use the FPTAS for E-$K$KP to obtain solution $S_{\omega}$.  Among all the solutions obtained, the final solution $S^{*}$ is chosen as the set with smallest objective value among all the $S_{\omega}$. In this section we focus on the scenario when the following assumption holds.
\begin{assumption}\label{assumption}
$\max_{e\in E}g(e)/\min_{e\in E}g(e)=$\poly$(n)$
\end{assumption}
Indeed, $[L^{\natural}, U^{\natural}]$ can be replaced by any interval $[L, U]$, such that $U/L=$poly$(n)$ and $g(O^{*})\in [L, U]$.


%
%
%
	
\paragraph{Performance analysis.}
Consider the iteration when parameter $\omega=\omega^{*}$ satisfies that 
\begin{align*}
{f^{\natural}}^{2}(O^{*})\in [(1-\varepsilon)\cdot \omega^{*}, \omega^{*}],
\end{align*}
the corresponding solution $S_{\omega^{*}}$ returned by the FPTAS satisfies that
\begin{align}\label{ineqg}
\bar{g}(S_{\omega^{*}})\geq (1-\varepsilon) \cdot \bar{g}(O^{*}),
\end{align}
and $g(S_{\omega^{*}})=C-\bar{g}(S_{\omega^{*}})\leq g(O^{*})+\varepsilon \cdot \bar{g}(O^{*})$. Under Assumption \ref{assumption}, we can obtain $g(S_{\omega^{*}})\leq (1+\varepsilon) \cdot g(O^{*})$ by replacing $\varepsilon$ by $\varepsilon/$poly$(n)$. In addition,
\begin{align}\label{ineqf}
f(S_{\omega^{*}})\leq \alpha(n) \cdot f^{\natural}(S_{\omega^{*}}) \leq \alpha(n) \cdot \sqrt{\omega^{*}}\leq  \alpha(n) \cdot \frac{f^{\natural}(O^{*}) }{\sqrt{1-\varepsilon}}\leq \frac{\alpha(n)}{\sqrt{1-\varepsilon}}\cdot f(O^{*})\leq  \frac{\alpha(n)}{1-\varepsilon}\cdot f(O^{*}).
\end{align} 
Consequently, we know that
\begin{align*}
\frac{F(O^{*})}{F(S^{*})}&\geq \frac{F(O^{*})}{F(S_{\omega^{*}})}
=\frac{f(O^{*})+g(O^{*})}{f(S_{\omega^{*}})+g(S_{\omega^{*}})} \overset{(a)}{\geq} \Big(1-\varepsilon \Big) \frac{f(S_{\omega^{*}})/\alpha(n)+g(S_{\omega^{*}})}{f(S_{\omega^{*}})+g(S_{\omega^{*}})}\\
&\overset{(b)}{\geq} \Big(1-\varepsilon \Big)\Big[1-\Big(1-\frac{1}{\alpha(n)}\Big)\frac{1}{1+\eta} \Big],
\end{align*}
where the first inequality follows from (\ref{ineqg}) and (\ref{ineqf}), parameter $\eta=\min_{S:S\subseteq E}{\{g(S)/f(S)\}}$ denotes the minimum ratio between $f$ and $g$, hence $g(S_{\omega^{*}})\geq \eta \cdot f(S_{\omega^{*}})$ and $(b)$ follows. Intuitively the difficulty of problem~(\ref{subminpro}) is related to $\eta$. More specifically, when $\eta$ increases, the problem becomes easier since the proportion of the modular function increases.
\paragraph{Approximation ratio lower bound.} For problem (\ref{subminpro}), the following approximation ratio lower bound is implied by~\citep{goemans2009approximating,svitkina2008submodular,iyer2013submodular}.
\begin{proposition}[\citep{goemans2009approximating,svitkina2008submodular,iyer2013submodular}]\label{sublowerbound}
Given a submodular function $f$ and modular function $g$, no polynomial time algorithm can achieve approximation ratio better than $1-\Big(1-\sqrt{\frac{\log n}{n}} \Big)\frac{1}{1+\eta}$, where $\eta=\min_{S:S\subseteq E}{\{g(S)/f(S)\}}$.
\end{proposition}

\end{document}